\documentclass[11pt]{article}

\usepackage[utf8]{inputenc}
\usepackage{amsmath,amssymb,amsthm}
\usepackage{geometry}
\usepackage{graphicx}
\usepackage{thmtools, thm-restate}
\usepackage{listings}
\usepackage[framemethod=tikz]{mdframed}
\usepackage[sort]{natbib}
\usepackage{float}
\usepackage{makecell}
\usepackage{cancel}
\usepackage[shortlabels]{enumitem}
\usepackage{bm}
\usepackage{caption, subcaption}
\usepackage{booktabs}  % For prettier table rules
\usepackage{array}     % For better column formatting

\usepackage{xcolor}
\usepackage[colorlinks=true,linkcolor=blue,citecolor=blue,urlcolor=blue]{hyperref}

\allowdisplaybreaks

\newtheorem{theorem}{Theorem}
\newtheorem{lemma}{Lemma}
\newtheorem{corollary}{Corollary}
\newtheorem{proposition}{Proposition}

\theoremstyle{definition}

\newtheorem{property}{Property}
\newtheorem{example}{Example}

\newtheorem{criterion}{Criterion}
\theoremstyle{remark}
\newtheorem{assumption}{Assumption}
\newtheorem{remark}{Remark}

\def\inprob{\stackrel{p}{\rightarrow}}
\def\indist{\rightsquigarrow}
\newcommand\ind{\protect\mathpalette{\protect\independenT}{\perp}}
\def\independenT#1#2{\mathrel{\rlap{$#1#2$}\mkern4mu{#1#2}}}

\DeclareSymbolFont{bbold}{U}{bbold}{m}{n}
\DeclareSymbolFontAlphabet{\mathbbold}{bbold}
\newcommand{\one}{\mathbbold{1}}
\def\cov{\text{cov}}
\def\bbE{\mathbb{E}}
\def\bbP{\mathbb{P}}
\def\bbR{\mathbb{R}}
\def\bbV{\mathbb{V}}

\title{\textbf{Comparing causal parameters with many treatments and positivity violations}}
\date{}
\author{\\ Alec McClean, Yiting Li, Sunjae Bae, Mara A. McAdams-DeMarco, \\
Iv\'{a}n D\'{i}az$^*$, Wenbo Wu$^*$ \vspace{0.05in} \\
Division of Biostatistics and Department of Surgery \\
NYU Grossman School of Medicine \vspace{0.05in} \\
\texttt{hadera01@nyu.edu,} \\
\texttt{\{yiting.li, sunjae.bae, mara.mcadamsdemarco\}@nyulangone.org,} \\ \texttt{ivan.diaz@nyu.edu, wenbo.wu@med.nyu.edu}}

\begin{document}
\maketitle

\begin{abstract}
    Comparing outcomes across treatments is essential in medicine and public policy. To do so, researchers typically estimate a set of parameters, possibly counterfactual, with each targeting a different treatment. Treatment-specific means are commonly used, but their identification requires a positivity assumption, that every subject has a non-zero probability of receiving each treatment. This is often implausible, especially when treatment can take many values. Causal parameters based on dynamic stochastic interventions offer robustness to positivity violations. However, comparing these parameters may fail to reflect the effects of the underlying target treatments because the parameters can depend on outcomes under non-target treatments. To clarify when two parameters targeting different treatments yield a useful comparison of treatment efficacy, we propose a comparability criterion: if the conditional treatment-specific mean for one treatment is greater than that for another, then the corresponding causal parameter should also be greater.  Many standard parameters fail to satisfy this criterion, but we show that only a mild positivity assumption is needed to identify parameters that yield useful comparisons. We then provide two simple examples that satisfy this criterion and are identifiable under the milder positivity assumption: trimmed and smooth trimmed treatment-specific means with multi-valued treatments. For smooth trimmed treatment-specific means, we develop doubly robust-style estimators that attain parametric convergence rates under nonparametric conditions. We illustrate our methods with an analysis of dialysis providers in New York State.
\end{abstract}


\vfill 

\renewcommand{\thefootnote}{\fnsymbol{footnote}} % Change numbering to symbols
\footnotetext[1]{Equal contribution.} % Use [1] for the first footnote, which will now be a "*".

\section{Introduction} \label{sec:intro}

Comparing a large set of treatments is a longstanding problem in causal inference which appears in various scientific fields. In public policy, researchers often evaluate a wide range of different policies or interventions, such as training programs with varying hours of training or educational interventions with varying classroom sizes \citep{card2010active, draper2004statistical, rubin2004potential}. Similar challenges also arise in medicine and public health. This paper is inspired by the problem of provider profiling, which is concerned with comparing the performance of potentially thousands of healthcare providers in terms of patient outcomes \citep{normand1997statistical}.

\bigskip

To compare treatments, researchers typically estimate a set of parameters, with each parameter targeting the effect of a different treatment, often under counterfactual conditions \citep{longford2020performance}.  A canonical example is the treatment-specific mean, the population counterfactual average outcome if all subjects received a specific treatment. Treatment-specific means are a natural choice because comparing two treatment-specific means describes how the same set of subjects (all subjects) would fare receiving two different treatments, and does not depend on subjects' outcomes under any other treatments.  Additionally, methods for efficiently estimating sets of treatment-specific means under standard causal assumptions are well established \citep{cattaneo2010efficient}. However, identifying treatment-specific means as functions of the observed data requires a positivity assumption, that all subjects have a non-zero chance of receiving every treatment.  When treatment can take many values, this assumption is often untenable. For example, in provider profiling there could be millions of patients attending thousands of providers across large geographic areas, and it is implausible that patients would attend distant providers. Thus, it is necessary to develop methods that are robust to positivity violations and preserve the desirable properties of treatment-specific means.

\bigskip

Dynamic stochastic interventions, which characterize outcomes under counterfactual shifts in the treatment distribution, can be robust to positivity violations. In the causal inference literature, these are well-studied for binary and continuous treatments, and recent research has proposed several options, including modified treatment policies, multiplicative shifts, and exponential tilts \citep{diaz2020causal,  diaz2012population, haneuse2013estimation, kennedy2019nonparametric, wen2023intervention}.  In provider profiling, ``indirectly standardized'' parameters are often used to compare provider efficacy \citep{kitagawa1955components}. These parameters, such as the standardized mortality ratio (SMR), consider counterfactual outcomes if the patients treated by one provider were instead treated according to a random draw from other providers.

\bigskip

Previously proposed dynamic stochastic interventions do not yield useful comparisons of treatment efficacy. With indirect standardization, differences in parameters targeting two treatments could be driven by differences in the covariate distributions receiving target treatments (e.g., \citet{marang2014simpson, shahian2008comparison}). Modern dynamic stochastic interventions from causal inference are not obviously susceptible to this issue, but suffer a second drawback: without careful construction, they can implicitly alter the covariate distribution across non-target treatments. Then, comparing two treatment-specific interventions may not reflect whether one treatment is better than another because differences in counterfactual parameters could be driven by how subjects fare under other, unrelated treatments.

\bigskip

Recently, \citet{roessler2021can} provided formal guidance on comparing provider performance. They proposed five axioms that the standardized mortality ratio ought to satisfy to enable comparisons of providers, and concluded that the standardized mortality ratio satisfied two. While their work is a critical contribution, it has several limitations. Most notably, their analysis does not easily generalize to other interventions in causal inference, because it was tailored to the standardized mortality ratio and did not use counterfactual quantities. Additionally, they did not suggest new interventions that might satisfy their axioms. 

\bigskip

To our knowledge, beyond \citet{roessler2021can} there is no further formal guidance on comparing treatment efficacy under positivity violations. This paper addresses this gap in the literature. First, we build on the work in \citet{roessler2021can} by proposing one simple comparability criterion grounded directly in counterfactual quantities.  Our criterion stipulates that if one treatment's conditional treatment-specific mean is almost surely larger than another's, then the causal parameter targeting the first treatment is larger, and if the conditional treatment-specific means are almost surely equal, then the causal parameters are equal. This criterion offers a straightforward way to evaluate whether sets of causal parameters allow us to compare treatment efficacy. We also illustrate that several typical estimands in the literature fail to satisfy this criterion.

\bigskip

To understand what interventions could satisfy the comparability criterion, we establish that a minimum positivity assumption is necessary to identify sets of comparable parameters, and this minimum positivity assumption can be milder than the usual positivity assumption. These results provide intuition for constructing comparable parameters and formalize the trade-off between the mildness of the positivity assumption and desirability of the comparability criterion: interventions satisfying a more desirable comparability criterion also require a stronger positivity assumption for identification.   

\bigskip

While other estimands in the literature either fail the comparability criterion or become unidentifiable under positivity violations, we propose two simple examples that overcome both challenges: trimmed and smooth trimmed treatment-specific means. These approaches trim across all propensity score values for multi-valued treatments and can selectively adapt to positivity violations across only a subset of covariates. They are part of a larger family of effects satisfying both criteria, examined in the appendix.

\bigskip

We then develop doubly robust-style estimators for the smooth parameters. Crucially, these estimators achieve parametric convergence rates and attain a normal limiting distribution under nonparametric conditions on estimators for the generalized propensity score and outcome regression. These methods can be used to efficiently estimate comparable causal parameters with any multi-valued treatments. Moreover, they are directly applicable to comparing healthcare providers, offering a novel method for provider profiling that is robust to positivity violations. To illustrate our methods, we analyze the performance of dialysis facilities in New York State.

\subsection{Structure of the paper}

Section~\ref{sec:setup} formally describes the data, mathematical notation, causal assumptions, and static deterministic and dynamic stochastic interventions. And, it provides further details on comparing treatments with treatment-specific means. Section~\ref{sec:conditions} defines the comparability criterion, illustrates common examples that fail to satisfy the criterion, and establishes that a minimum positivity assumption is necessary to satisfy the comparability criterion and identifiability simultaneously.  Section~\ref{sec:examples} presents examples --- trimmed and smooth trimmed effects --- that comply with the comparability criterion and maintain identifiability under positivity violations. Section~\ref{sec:eif-estimator} identifies these parameters and develops doubly robust-style estimators for smooth trimmed effects, leveraging nonparametric efficiency theory. Section~\ref{sec:data-analysis} illustrates these methods with an analysis of claims data from dialysis providers in New York State, Section~\ref{sec:simulations} illustrates these methods with a simulation study, and Section~\ref{sec:discussion} concludes and discusses future work.

\section{Setup and background} \label{sec:setup}

\subsection{Data, parameters, and nuisance functions}

We assume we observe $n$ observations drawn iid from some distribution $\bbP$ in a space of distributions $\mathcal{P}$. In other words, we observe data $\{ Z_i \}_{i=1}^{n} = \{ (X_i, A_i, Y_i) \}_{i=1}^{n} \stackrel{iid}{\sim} \bbP \in \mathcal{P}$, where $X \in \bbR^p$ are $p$-dimensional covariates, $A \in \{1, \dots, d\}$ is a categorical treatment with $d$ levels, and $Y \in \mathcal{Y} \subseteq \bbR$ is an outcome.  We refer to two nuisance functions: let $\pi_a(X)= \bbP(A = a \mid X)$ denote the generalized propensity score, the probability of receiving treatment $a$ for subjects with covariates $X$, and let $\mu_a(X) = \bbE(Y \mid A =a, X)$ denote the conditional mean outcome supposing subjects with covariates $X$ received treatment $a$. We denote the potential outcome supposing treatment $a$ was received by $Y^a$.  Implicitly, we suppose the $n$ observations also correspond to unobserved samples drawn from a counterfactual distribution $\bbP^c$ which includes all potential outcomes; i.e., $\{ X_i, A_i, Y_i^1, \dots, Y_i^d \} \stackrel{iid}{\sim} \bbP^c$. Unless necessary for clarification, we omit $c$ superscripts, with the understanding that counterfactual objects are clear from context (e.g., if they involve potential outcomes). We define a causal/counterfactual parameter $\psi$ as a map from the counterfactual distribution to the reals; i.e., $\psi: \bbP^c \to \bbR$. Under certain identifying assumptions, established subsequently, $\psi = \psi_{obs}$, where $\psi_{obs}: \bbP \to \bbR$ is a map from the observational distribution to the reals.

\subsection{Mathematical notation}

For a function $f(Z)$, we use $\lVert f \rVert = \sqrt{\int f(z)^2 d\bbP(z)}$ to denote the $L_2(\bbP)$ norm, $\bbP(f) = \int_{\mathcal{Z}} f(z) d\bbP(z)$ to denote the average with respect to the underlying distribution $\bbP$, and $\bbP_n(f) = \frac{1}{n} \sum_{i=1}^{n} f(Z_i)$ to denote the empirical average with respect to the $n$ observations.  In a standard abuse of notation, when $A$ is an event we let $\bbP(A)$ denote the probability of $A$.  We also denote expectation and variance with respect to the underlying distribution by $\bbE$ and $\bbV$, respectively.  We use the notation $a \lesssim b$ to mean $a \leq Cb$ for some constant $C$, $\indist$ to denote convergence in distribution, and $\inprob$ for convergence in probability to zero. Additionally, we use $o_\bbP(\cdot)$ to denote convergence in probability, i.e., if $X_n$ is a sequence of random variables then $X_n = o_\bbP (r_n)$ implies $\left| \frac{X_n}{r_n} \right| \inprob 0$.

\subsection{Causal assumptions}

In provider profiling and many other applications, observational data is typically the only available data. Several causal assumptions are necessary to identify counterfactual parameters as functionals of the observed data. Two standard assumptions are consistency and exchangeability.
\begin{assumption}
    \emph{Consistency: } $A =a \implies Y = Y^a$.
\end{assumption}

\begin{assumption}
    \emph{Exchangeability: } $Y^a \ind A \mid X$ for all $a \in \{1, \dots, d\}$.
\end{assumption}

Consistency asserts that we observe the potential outcome $Y^a$ relevant to the observed treatment $a$. Consistency would be violated if, for example, there were interference between subjects, such that one subject's treatment choice affected another's outcome.  Exchangeability says that the treatment is as-if randomized within covariate strata; in other words, there are no unobserved confounders which might affect subjects' treatment choice and their subsequent outcomes.  It would be violated if there existed an unmeasured confounder that predicted both treatment receipt and subsequent outcomes.  The literature addressing violations of each of these assumptions is too large to summarize here, but see, for example, \citet{tchetgen2012causal} and \citet{richardson2014nonparametric}.  The final typical causal assumption is positivity.  The ``strong'' positivity assumption asserts that every subject's probability of receiving each treatment is bounded away from zero, while the ``weak'' positivity assumption dictates these probabilities  are non-zero.
\begin{assumption} \label{asmp:strong-positivity}
    \emph{Strong positivity: } $\exists\ \epsilon > 0$ s.t. $\bbP \{ \epsilon \leq \pi_a(X)  \} = 1\ \forall\ a \in \{1, \dots, d\}.$
\end{assumption}

\begin{assumption} \label{asmp:weak-positivity}
    \emph{Weak positivity: } $\bbP \{ 0 < \pi_a(X) \} = 1 \ \forall\ a \in \{1, \dots, d \}$.    
\end{assumption}
Weak positivity, along with consistency and exchangeability, is sufficient to identify the causal parameters we consider as functionals of the observed data, while strong positivity is sufficient to construct semiparametric efficient estimators \citep{khan2010irregular}. However, often weak positivity is untenable when there are many treatment levels. For example, given the vast number of providers and patients, weak positivity is unrealistic in provider profiling.  Therefore,  we will develop methods for comparing causal parameters under milder positivity assumptions.

\subsection{Comparisons with treatment-specific means}

Before discussing interventions that can be robust to positivity violations, we first discuss why treatment-specific means are natural causal parameters for comparing treatments and why they are less useful under positivity violations. The treatment-specific mean targeting treatment $a$ is $\bbE \left(Y^a \right)$, the counterfactual mean outcome if all subjects received treatment $a$. Comparing two treatment-specific means yields a meaningful comparison of treatment efficacy since differences in $\bbE \left(Y^a\right)$ and $\bbE\left(Y^b\right)$ arise solely from subjects’ outcomes under treatments $a$ and $b$. These parameters indicate how treatments $a$ and $b$ would affect identical subject populations (all subjects) when all other treatments are received by identical subject populations (no subjects). Moreover, these parameters are invariant to the distribution $d\bbP(X \mid A=a)$ of covariates among the subjects receiving each treatment in the observed data.

\bigskip

Supposing exchangeability, consistency, and strong positivity hold, then the set of treatment-specific means can be identified by $\big\{ \bbE(Y^a) \big\}_{a=1}^{d} = \Big\{ \bbE \{ \mu_a(X) \} \Big\}_{a=1}^{d},$ where $\bbE \{ \mu_a(X)\}$ is based on observed data and can be estimated from it. Therefore, if the three causal assumptions held, one could compare treatments by estimating the set of treatment-specific means from the observed data. However, violations of weak positivity make it impossible to identify and estimate the set of treatment-specific means from observed data. Moreover, even when strong positivity holds, estimating $\big\{ \bbE\{ \mu_a(X) \} \big\}_{a=1}^{d}$ can be challenging. The nonparametric efficiency bound for estimating $\bbE \{ \mu_a(X) \}$ increases dramatically when \( \pi_a(X) \approx 0 \), leading to high variance estimates and wide confidence intervals, which can hinder meaningful conclusions from the analysis \citep{hahn1998role, van2000asymptotic}.

\subsection{Static deterministic and dynamic stochastic interventions}

Under positivity violations, researchers typically consider causal parameters defined by dynamic stochastic interventions  instead of static deterministic interventions like treatment-specific means. Here, we formally define these terms and then discuss an example that is robust to positivity violations with binary treatment. A stochastic intervention randomly assigns subjects to treatments according to a draw from a distribution. With multi-valued treatments, this is a categorical distribution (with binary treatment, this is a Bernoulli distribution). Mathematically, the parameters of interest are $\bbE(Y^{Q})$ where $Q \sim \text{Categorical}\big\{q(A=1), \dots, q(A=d) \big\}$ and $q(A=i)$ is the probability of receiving treatment $i$ under the intervention. Stochastic interventions allow subjects to receive one of several treatments based on a random draw. By contrast, deterministic interventions assign subjects to treatments deterministically, so there is no randomness in counterfactual treatment receipt. Dynamic interventions let treatment receipt probabilities vary with covariates. Mathematically, the parameters take the form $\bbE\left(Y^{Q}\right)$, but where $Q \sim \text{Categorical}\big\{ q(A=1 \mid V), \dots, q(A=d \mid V) \big\}$ and $q(A=i  \mid V)$ can vary with $V \subseteq X$. 

\bigskip

The incremental propensity score intervention (IPSI) is a popular dynamic stochastic intervention that is robust to positivity violations \citep{kennedy2019nonparametric, bonvini2023incremental}. It considers the counterfactual outcomes when the odds of a binary treatment are multiplied by some factor. Mathematically, for binary treatments, it is $\bbE \left\{ Y^{Q(\delta)} \right\}$ for $\delta \in (0, \infty)$, where $Q(\delta) \sim \text{Bernoulli} \Big( q\{ \pi_1(X); \delta\} \Big)$ and $q(x; \delta) = \frac{\delta x}{\delta x + 1 - x}$.  IPSIs remain identifiable under violations of weak positivity because they consider counterfactual interventions where subjects who would always receive treatment do so, and those who would never receive treatment remain untreated. In other words, when the propensity scores equal zero or one, so do the interventional propensity scores (\( q(0; \delta) = 0, q(1; \delta) = 1\)). Hence, whenever the outcome regressions are undefined due to weak positivity violations, the interventional propensity score is zero and the resulting causal parameter is still identifiable.  Under only consistency and exchangeability, $\bbE \left\{ Y^{Q(\delta)} \right\} = \bbE \left( q\{ \pi_1(X); \delta \} \mu_1(X) + \big[ 1 - q\{ \pi_1(X); \delta \} \big] \mu_0(X)\right)$. IPSIs are a special case of exponential tilts \citep{diaz2020causal, schindl2024incremental}. For continuous treatment, the interventional propensity score with exponential tilts is $q(A=a \mid X; \delta) = \frac{\exp(\delta a) \pi_a(X)}{\int_{\mathcal{A}} \exp (\delta a) \pi_a(X) dp(a)}.$

\section{Comparability and positivity} \label{sec:conditions}

With many treatments, the crucial tradeoff is between the goal of comparing treatment efficacy across all treatments, and the constraint of positivity violations. In this section, we propose a novel criterion that formalizes this goal. We then provide further intuition for the criterion, and, importantly, show that standard estimands proposed in the literature fail to satisfy it. Finally, we relate the criterion to the constraint of positivity violations and establish that a minimum positivity assumption is necessary to identify causal parameters that satisfy this criterion. In the next section, we introduce two examples that satisfy the criterion and are identifiable under only the minimum positivity assumption derived in this section.

\subsection{Comparability criterion} 

When comparing two counterfactual parameters $\psi_a$ and $\psi_b$ targeting treatments $a$ and $b$, we say these parameters are ``$V$-comparable'' (for $V \subseteq X$) if they satisfy the following condition.

\begin{criterion} \label{crit:comp}
    \textbf{($V$-comparability)} For a covariate subset $V \subseteq X$, the parameters $\psi_a$ and $\psi_b$ are $V$-comparable if they satisfy:
    \begin{align}
        &\bbP^c \{ \bbE(Y^a \mid V) > \bbE(Y^b \mid V) \} = 1 \implies \psi_a > \psi_b, \label{eq:cond_lower} \\
        &\bbP^c \{ \bbE(Y^a \mid V) = \bbE(Y^b \mid V) \} = 1 \implies \psi_a = \psi_b, \text{ and } \label{eq:cond_equal} \\
        &\bbP^c \{ \bbE(Y^a \mid V) < \bbE(Y^b \mid V) \} = 1 \implies \psi_a < \psi_b \label{eq:cond_upper} 
    \end{align}
    for all counterfactual distributions $\bbP^c$ in the counterfactual model.
\end{criterion}

This criterion asks that the relationship between $\psi_a$ and $\psi_b$ reflects the relationship between $\bbE(Y^a \mid V)$ and $\bbE(Y^b \mid V)$: if one conditional mean is almost surely larger, then the corresponding causal parameter is larger; if they are almost surely equal, the causal parameters are equal.   Moreover, it asks that this relationship is independent of potential outcome means under treatments other than $a$ and $b$.  

\medskip

The criterion is not an assumption; rather, it is a desideratum. If a set of parameters satisfy the criterion and the left-hand sides of \eqref{eq:cond_lower}-\eqref{eq:cond_upper} hold (which is an assumption), then the right-hand sides of \eqref{eq:cond_lower}-\eqref{eq:cond_upper} will follow.  Importantly, the set of treatment-specific means $\left\{ \bbE(Y^a) \right\}_{a=1}^{d}$ satisfies criterion~\ref{crit:comp} for all $V \subseteq X$. Therefore, if criterion~\ref{crit:comp} holds for another set of parameters $\{\psi_a\}_{a=1}^{d}$, and the left-hand sides of \eqref{eq:cond_lower}-\eqref{eq:cond_upper} hold, then the order of the parameters $\{\psi_a\}_{a=1}^{d}$ matches the order of the treatment-specific means. 

\medskip

This condition builds on work in \citet{roessler2021can}, who proposed five axioms for analyzing whether the standardized mortality ratio allows for useful comparisons of treatment efficacy.  We propose instead one, defined in terms of counterfactual quantities. This is more straightforward to assess for many estimands of interest. Since researchers typically focus on estimating weighted averages of potential outcomes, concentrating on conditional potential outcomes directly appears to be the most natural way to define comparability, which is why we adopt it here. However, alternative notions of comparability could be considered. For instance, one might define comparability directly in terms of the potential outcomes instead of their conditional means, and use antecedents such as $\bbP(Y^a = Y^b) = 1$ or $\bbP(Y^a < y) < \bbP(Y^b < y)$ for all $y \in \mathcal{Y}$ in Criterion~\ref{crit:comp}. We found that, when establishing results for what parameters can satisfy the criterion, these alternative notions of comparability are only useful insofar as they imply Criterion~\ref{crit:comp}.  Therefore, we focus on Criterion~\ref{crit:comp} as the most natural notion of comparability. 

\bigskip

To understand when the criterion is more or less desirable as $V$ changes, we can examine \eqref{eq:cond_lower}-\eqref{eq:cond_upper} directly or examine their contrapositives. First, notice that the left-hand sides of \eqref{eq:cond_lower}-\eqref{eq:cond_upper} become stronger as $V$ grows larger, being strongest for $V = X$ and weakest for $V = \emptyset$. Therefore, the parameters satisfy the most (least) desirable comparability criterion when $V = \emptyset$ ($V = X$) because they require the least (most) stringent assumption to guarantee comparability.  Alternatively, consider the contrapositives of each. For example, the contrapositive of $\eqref{eq:cond_lower}$ yields 
\[
\psi_a \leq \psi_b \implies \bbP^c \{ \bbE(Y^a \mid V) > \bbE(Y^b \mid V) \} < 1.
\]
For $V = X$, knowing that $\psi_a \leq \psi_b$ provides little information about the conditional treatment-specific means or the treatment-specific means themselves. By contrast, when $V = \emptyset$, $\psi_a \leq \psi_b \implies \bbE(Y^a) \leq \bbE(Y^b)$, a much stronger statement.  Indeed, combining the contrapositives of \eqref{eq:cond_lower}-\eqref{eq:cond_upper} for $V = \emptyset$ yields equivalence statements, such as:
\[
\psi_a > \psi_b \iff \bbE(Y^a) > \bbE(Y^b).
\]
In other words, if a set of parameters satisfy the $\emptyset$-comparability criterion, one can surmise the order of the treatment-specific means from the order of $\{ \psi_a \}_{a=1}^{d}$, without further assumptions.

\begin{remark} \label{rem:continuous}
    The $V$-comparability criterion naturally extends to continuous treatments when comparing different treatment levels. Consider a continuous treatment $A \in \mathcal{A} \subseteq \bbR$. The criterion can apply directly to comparing two different levels of the treatment. The points on the dose-response curve at two treatments of interest satisfy the criterion. 
\end{remark}

\subsection{Examples}

To illustrate why the $V$-comparability criterion is useful, we review three other proposals in the literature. We show that standard proposals either are unidentifiable under positivity violations or fail to satisfy the $V$-comparability criterion.

\medskip

Treatment-specific means are the canonical estimand for comparing treatment efficacy. These satisfy $V$-comparability for all covariate subsets but are unidentifiable under any positivity violations. This limitation has motivated two alternatives.

\medskip

Indirect standardization is the typical solution to positivity violations in provider profiling. One indirectly standardized parameter targeting treatment $a$ is the standardized mortality ratio $\tfrac{\bbE (Y^\Pi \mid A=a)}{\bbE(Y \mid A=a)}$, where $\Pi \sim \text{Categorical}\{\pi_1(X), \dots, \pi_d(X) \}$. This parameter compares the outcomes of patients under treatment $a$, $\bbE(Y \mid A=a)$, to their counterfactual outcomes had they been randomly reassigned to another treatment according to their generalized propensity score, $\bbE(Y^\Pi \mid A=a)$ \citep{daignault2017doubly}. While this parameter is always identifiable, it fails to satisfy $V$-comparability. Intuitively, this occurs because the parameters condition on actual treatment receipt, so different parameters consider different patient populations. This means that differences in standardized mortality ratios may not reflect the relative efficacy of the treatments under consideration. This phenomenon is well-documented in the literature, and our framework adds additional formal justification for why standardized mortality ratios are not useful for comparing treatment efficacy \citep{shahian2008comparison, shahian2020use, manktelow2014differences, marang2014simpson}. 

\medskip

Dynamic stochastic interventions are another approach. A naive dynamic stochastic intervention might tilt the intervention propensity score towards a target treatment and decrease all other probabilities equally. For example, 
\[
Q_1 \sim \text{Categorical} \left\{ f\{ \pi_1(X) \}, \pi_2(X) \big[ 1 - f\{ \pi_1(X) \} \big], \dots, \pi_d(X) \big[ 1 - f\{ \pi_1(X) \} \big] \right\},
\] 
where $f(\cdot)$ represents some upward shift, such as an exponential tilt or multiplicative shift. While these interventions can be identifiable under positivity violations, they fail to satisfy $V$-comparability.  Intuitively, this occurs because they implicitly intervene on non-target treatment probabilities in different ways across interventions targeting different treatments. Consider a second intervention, targeting treatment $a=2$, which we will compare to the intervention above targeting treatment $a=1$:
\[
Q_2 \sim \text{Categorical}\left\{ \pi_1(X) \big[ 1 - f\{ \pi_2(X) \} \big], f\{ \pi_2(X)\} , \dots, \pi_d(X) \big[ 1 - f\{ \pi_2(X) \} \big] \right\}.
\]
If $\pi_1(X) \neq \pi_2(X)$, the interventions ascribe different probabilities of receiving non-target treatments (e.g., $\pi_d(X) \big[ 1 - f\{ \pi_1(X) \} \big] \neq \pi_d(X) \big[ 1 - f\{ \pi_2(X) \} \big]$ unless $f(\cdot)$ is a constant), and $\bbE(Y^{Q_1}) - \bbE(Y^{Q_2})$ could depend on subjects' outcomes under non-target treatments for $a > 2$.

\subsection{The necessary positivity assumption}

The prior sections proposed our new $V$-comparability criterion, justified it, and showed that standard parameters in the literature fail to satisfy $V$-comparability, or are unidentifiable under positivity violations. In the next section, we'll propose examples that satisfy $V$-comparability and are identifiable under positivity violations. Before that, here we formally establish the tradeoff between the $V$-comparability and positivity, by proving there is a necessary minimum positivity assumption that must hold in order for $V$-comparable parameters to exist, in the following result.

\begin{theorem} \label{thm:positivity} \textbf{(Necessary positivity assumption)}
    Let $\big\{ \psi_a \big\}_{a=1}^{d} = \big\{ \bbE(Y^{Q_a}) \big\}_{a=1}^{d}$ denote a set of parameters defined by dynamic stochastic interventions that vary with covariates $V \subseteq X$ and target treatments $1, \dots, d$, respectively.  Moreover, let
    \begin{equation} \label{eq:trimmed-set}
        C_V = \left\{ v: \bbP \left\{ \pi_a(X) > 0 \mid V = v \right\} = 1 \ \forall \ a \in \{1, \dots, d\} \right\}
    \end{equation}
    denote the set of subjects who have a non-zero probability of receiving every treatment.  If the parameters satisfy $V$-comparability (criterion~\ref{crit:comp}) and the parameters are identifiable, then $\bbP(C_V) > 0$; i.e., $C_V$ must have positive probability.  In other words, $\bbP(C_V) > 0$ is necessary for identifiability and $V$-comparability to hold simultaneously. 
\end{theorem} 

All proofs are delayed to the appendix. Theorem~\ref{thm:positivity} establishes the necessary positivity assumption for identifying $V$-comparable parameters; it asserts that there must exist a set of subjects across covariates $V$ that have a positive probability of receiving every treatment. Crucially, the result applies to any parameters that would satisfy $V$-comparability and be identifiable; the specific construction of the parameters does not matter. Stated equivalently, the result asserts that 
\begin{equation} \label{eq:positivity-condition}
    \bbE \left( \prod_{a=1}^{d} \one \Big[  \bbP \left\{ \pi_a(X) > 0 \mid V \right\} = 1 \Big]  \right) > 0
\end{equation}
is necessary for $\{ \psi_a \}_{a=1}^{d}$ to be $V$-comparable and identifiable. Although \eqref{eq:positivity-condition} is a positivity assumption, it can be considerably weaker than the weak positivity assumption in Assumption~\ref{asmp:weak-positivity}.   When $\emptyset \subset V \subseteq X$, \eqref{eq:positivity-condition} is only an intermediate assumption. For example, when $V=X$,~\eqref{eq:positivity-condition} simplifies to $\bbE \left[ \prod_{a=1}^{d} \one \Big\{ \pi_a(X) > 0 \Big\} \right] > 0$. Meanwhile, when $V = \emptyset$, then \eqref{eq:positivity-condition} simplifies to weak positivity.

\subsection{Comparability versus positivity}

Theorem~\ref{thm:positivity} formalizes the tension between satisfying the $V$-comparability criterion and identifiability simultaneously. Dynamic interventions that vary with a more granular covariate set $V$ only require a milder positivity assumption, but can only guarantee a less desirable $V$-comparability criterion.  Meanwhile, interventions that vary with a coarser covariate set can guarantee a more desirable comparability criterion, but require a stronger positivity assumption.  For example, $\emptyset$-comparability is the most desirable comparability criterion, but parameters satisfying $\emptyset$-comparability require weak positivity for identification.  By contrast, one can construct dynamic stochastic interventions varying with $X$ that only require a mild positivity assumption for identification, but can only satisfy the $X$-comparability criterion, which is the least desirable criterion.

\medskip

One might hope that dynamic interventions which vary with $X$ and only require a mild positivity assumption could satisfy a more desirable comparability criterion for $V \subset X$. However, this is not possible. Consider two parameters $\psi_a$ and $\psi_b$ corresponding to dynamic interventions varying over $X$, and suppose we hoped they would satisfy the $V$-comparability criterion for $V \subset X$. Because the interventions vary over with $X$, even if $\bbE(Y^a \mid V) > \bbE(Y^b \mid V)$ almost surely, it is impossible to rule out pathological conditional mean distributions $\bbE(Y^a \mid X)$ and $\bbE(Y^b \mid X)$ which yield $\psi_a \leq \psi_b$.

\section{Comparable and identifiable examples: trimmed and smooth trimmed treatment-specific means} \label{sec:examples}

In the appendix, we develop a general framework for constructing dynamic stochastic interventions that satisfy the $V$-comparability criterion and yield identifiable parameters.  The framework is based on two intuitive properties: (i) an intervention targeting treatment $a$ should increase the probability of receiving treatment $a$ and decrease the probability of receiving other treatments, and (ii) two interventions targeting different treatments should have the same interventional propensity score at all other non-target treatments. 

\medskip

Here, we focus on two specific examples: trimmed and smooth trimmed treatment-specific means. We focus on these parameters because they approximate the treatment-specific means as closely as possible while adapting to positivity violations. Theorem~\ref{thm:positivity} dictated the positivity assumption necessary for $V$-comparability and identifiability. Therefore, we construct the trimmed and smooth trimmed treatment-specific means to adapt to violations of that assumption. We define the $V$-comparable trimmed treatment-specific mean targeting treatment $a$ as the mean potential outcome under the following intervention:
\begin{align*}
    Q_a \sim \text{Categorical}\big\{ q_1(A=b \mid V), \dots, q_d(A=b \mid V) \big\} \text{ where } \\
    q_a(A=b \mid V) = \begin{cases} 
        \pi_b(V) \one(V \notin C_V) &\text{ for } b \neq a, \\
        \one(V \in C_V) + \pi_a(V) \one(V \notin C_V) &\text{ otherwise,} 
    \end{cases}
\end{align*}
where $C_V$ is the set of subjects with a non-zero probability of receiving every treatment, as defined in Theorem~\ref{thm:positivity}.  The trimmed treatment-specific mean satisfies the $V$-comparability criterion and is identifiable under the positivity assumption in \eqref{eq:positivity-condition} because it deliberately adapts to violations of that assumption: when $V \notin C_V$, then subjects are assigned treatment according to their underlying generalized propensity score, thereby ensuring the positivity assumption holds by construction.

\medskip

By focusing on the trimmed subjects in $C_V$, our construction agrees with prior intuition from the matching and balancing weights literature \citep{stuart2010matching, li2019propensity}.  Recent work in provider profiling has explored matching with multi-valued treatments \citep{silber2020comparing} and template matching, which matches each provider's patient population to a representative template \citep{silber2014template, vincent2021hospital}, while generalized matching methods for multi-valued treatment have also been developed \citep{yang2016propensity}. Meanwhile, \citet{li2019propensity} showed that balancing weights based on the generalized propensity score can be used to identify weighted average differences in conditional treatment-specific means with multi-valued treatments that are identifiable under positivity violations. With both matching and balancing weights, it has been argued informally that focusing on the group with non-zero probability of attending each treatment facilitates useful comparisons between treatments.\footnote{For example, \citeauthor{li2019propensity} argue that the ``generalized overlap weights focus on the subpopulation with substantial probabilities to be assigned to all treatments. This target population aligns with the spirit of randomized clinical trials by emphasizing patients at clinical equipoise, and is thus of natural relevance to medical and policy studies.''} Our construction deviates from prior approaches in that the trimming adapts to positivity violations of the milder necessary positivity assumption in Theorem~\ref{thm:positivity}, rather than weak positivity (Assumption~\ref{asmp:weak-positivity}). Our analysis provides a formal justification for why prior methods for weighting and matching yield parameters that can be used to compare treatment efficacy across treatments. 

\subsection{Smooth trimmed treatment-specific means}

It can be difficult to estimate trimmed treatment-specific means because it requires estimating the indicator function $\one(V \in C_V)$, which is non-differentiable and makes the resulting trimmed treatment-specific means non-smooth. As a result, standard semiparametric efficiency theory is inapplicable \citep{bickel1993efficient}. Achieving $\sqrt{n}$-efficiency under nonparametric assumptions can be done in two ways. First, one can target the data-dependent trimmed treatment-specific mean, which depends on the estimated indicator functions $\widehat \one (V \in C_V)$. Alternatively, one can smooth the indicator function. We take the second approach and consider smooth trimming indicator that is constructed with a smooth approximation of the trimming indicator $\one(V \in C_V)$. It is 
\begin{equation} \label{eq:smooth-trim}
    S(V \in C_V) = \prod_{b=1}^{d} \bbE \Big[ s \{ \pi_b(X) \} \mid V \Big],
\end{equation}
where $s(x)$ is any smooth approximation of $\one(x > 0)$. The formulation of $S(V \in C_V)$ is natural because $\one(V \in C_V)  = \one \left[ \prod_{b=1}^{d} \bbP \left\{ \pi_b(X) > 0 \mid V \right\} = 1 \right]$, which can be approximated by $\prod_{b=1}^{d} \bbP \left\{ \pi_b(X) > 0 \mid V \right\}$, and $\bbP \left\{ \pi_b(X) > 0 \mid V \right\} = \bbE \left[ \one \left\{ \pi_b(X) > 0 \right\} \mid V \right]$ can be approximated by $\bbE \left[ s\{ \pi_b(X) \} \mid V \right]$. 

\medskip

The smoothing function $s(x)$ can be chosen so that $S(V \in C_V)$ retains two important properties. First, $s(x)$ must approach $1$ very quickly as $x$ increases from zero. This is important because it ensures $S(V \in C_V)$ approximates $\one(V \in C_V)$ well. Second, $s(0) = 0$. This ensures that $S(V \in C_V) = 0$ if $V \notin C_V$; i.e., only subjects in $C_V$ are intervened upon. This second property is arguably more crucial because it guarantees that interventions using the smoothed indicator can satisfy the $V$-comparability criterion while still being identifiable under the necessary positivity assumption in~\eqref{eq:positivity-condition}. We only require $s(0) = 0$, and do not require any behavior for $s(x)$ when $x$ is negative because the inputs to $s(\cdot)$ are propensity scores, which are always non-negative.

\medskip

A simple example is $s(x) = 1- \exp(-kx)$ for $k > 0$. This function is smooth, satisfies $s(0) = 0$, and approaches $1$ quickly as $x$ increases from zero. Note that this function is a poor approximation of $\one(x > 0)$ if $x$ can be negative, but propensity scores are non-negative, and therefore this is not an issue.  We define the smooth trimmed treatment-specific mean targeting treatment $a$ as the mean potential outcome under a dynamic stochastic intervention $Q_a$ with the following interventional propensity scores:
\begin{equation} \label{eq:smooth-prop-scores}
    q_a(A=b \mid V) = \begin{cases}
        \{ 1 - S(V \in C_V) \} \pi_b(V) &\text{ for } b \neq a, \\
        S(V \in C_V) + \{ 1 - S(V \in C_V) \} \pi_a(V) &\text{ otherwise.}
    \end{cases}
\end{equation}
This construction weights the intervention towards treatment $a$ as the smooth trimming indicator approaches one, and weights the intervention towards the underlying generalized propensity scores as the smooth indicator approaches zero.

\begin{remark}
    Smooth approximations of functions of the propensity scores have been considered extensively in the trimming literature (see \citet{yang2018asymptotic, khan2022doubly, branson2023causal} for recent examples). As far as we are aware, the second property of our smooth approximation --- that $s(0) = 0$ --- is a novel constraint which leads to a different smooth approximation than those previously considered. Standard trimming methods develop smooth approximations of \(\one \left\{ \pi_a(X) > t \right\}\) for \(t > 0\), where \(t\) is the trimming threshold, and typically assume strong positivity holds. As a result, these smooth approximations are positive even when \(\pi_a(X) = 0\). For example, \citeauthor{yang2018asymptotic} and \citeauthor{branson2023causal} consider \(s(x) = \Phi_\epsilon\{ \pi_a(X) - t\}\), where \(\Phi_\epsilon\) is the CDF of a normal distribution with mean zero and variance \(\epsilon^2\). In our framework, this smooth approximation is undesirable because \(\Phi_\epsilon (z - t) > 0\) for all \(z \in [0,1]\) if \(t\) is finite. The resulting parameters will not be identifiable under the milder positivity assumption in \eqref{eq:positivity-condition}.  We instead consider $s(x) = 1 - \exp(-kx)$, which satisfies $s(0) = 0$ and retains identifiability under the milder positivity assumption.  Our approach may also be relevant to the trimming literature, as our construction appears to be new. 
\end{remark}

\section{Identification and estimation} \label{sec:eif-estimator}

In this section, we develop methods for estimating the set of smooth trimmed effects defined in the prior section. We begin with identification.
\begin{proposition} \label{prop:id} \textbf{(Identification)}
    Let $\big\{ \psi_a \big\}_{a=1}^{d} = \big\{ \bbE(Y^{Q_a}) \big\}_{a=1}^{d}$ denote the set of smooth trimmed treatment-specific means with interventional propensity scores as in \eqref{eq:smooth-prop-scores}.  Suppose consistency and exchangeability hold, and the positivity assumption in \eqref{eq:positivity-condition} holds for $V$. Then, for all $a \in \{1, \dots, d\}$,
    $$
    \psi_a = \bbE \left[ \sum_{b=1}^{d} \bbE \{ \mu_b(X) \mid V \} q_a(A =b \mid V) \right].
    $$
    where $q_a(A=b \mid V)$ is defined in \eqref{eq:smooth-prop-scores}.
\end{proposition}
This result is the typical g-formula applied to smooth trimmed treatment-specific means \citep{robins1986new}. The result suggests a natural plug-in estimator:
$$
\bbP_n \left\{ \sum_{b=1}^{d} \widehat \bbE \{ \widehat \mu_b(X) \mid V \} \widehat q_a(A = b \mid V) \right\},
$$
where $\widehat \mu_b(X)$ regresses $Y \sim \{ X,  \one(A=b) \}$, $\widehat \bbE \{ \widehat \mu_b(X) \mid V \}$ regresses $\widehat \mu_b(X) \sim V$, and $\widehat q_a(A=b \mid V)$ plugs estimated propensity scores into the definition of $q_a(A=b\mid V)$. With well-specified parametric models for the propensity score, outcome regression, and second-stage regression $\bbE \{ \mu_b(X) \mid V\}$, the plug-in estimator can achieve $\sqrt{n}$-convergence to $\psi_a$.  However, if the models are mis-specified, the plug-in estimator can be biased \citep{vansteelandt2012model}. Meanwhile, if the propensity score and outcome regression are estimated with nonparametric methods, the plug-in estimator will typically inherit slower-than-$\sqrt{n}$ nonparametric convergence rates. This motivates estimators based on nonparametric efficiency theory \citep{bickel1993efficient, van2000asymptotic, tsiatis2006semiparametric}. 

\subsection{Efficient influence function and one-step estimator}

The first-order bias of the nonparametric plug-in can be characterized by the efficient influence function (EIF) of the parameter, which can be thought of as the first derivative in a von Mises expansion of the parameter \citep{von1947asymptotic}. The EIF can be used to construct estimators that can achieve $\sqrt{n}$-convergence with nonparametric estimators for the nuisance functions.  In this section, we establish a doubly robust-style estimator based on the EIFs of the parameters $\big\{ \bbE(Y^{Q_a}) \big\}_{a=1}^{d}$.  We focus on $V = X$ for simplicity, but provide a comprehensive analysis in the supplement.  To derive the EIF and establish a doubly robust-style estimator, we require a strengthening of the necessary positivity assumption from~\eqref{eq:positivity-condition}.
\begin{assumption} \label{asmp:intermediate-positivity}
    \emph{Intermediate positivity: } $\exists\ \varepsilon > 0$ such that $\bbE \left[ \prod_{a=1}^{d} \one \Big\{ \pi_a(X) > 0 \Big\} \right] \geq \varepsilon.$ 
\end{assumption} 
This assumption is stronger than \eqref{eq:positivity-condition} in the same sense strong positivity is stronger than weak positivity: Assumption~\ref{asmp:intermediate-positivity} requires boundedness away from zero. It is necessary in the same way that typical strong positivity is necessary for establishing semiparametric efficient estimators \citep{khan2010irregular}. The next result provides the EIF for $\bbE(Y^{Q_a})$ with generic smooth approximation function $s(\cdot)$.

\begin{theorem} \label{thm:eif} \textbf{(Efficient influence function)} 
    Let $\psi_a$ denote the identified smooth trimmed treatment-specific mean in Proposition~\ref{prop:id} and let $V = X$. Suppose Assumption~\ref{asmp:intermediate-positivity} holds and $s(\cdot)$ is twice differentiable with non-zero bounded second derivatives. Then, the un-centered EIF of $\psi_a$ is
    \begin{align} 
        \varphi_a (Z) = \sum_{b=1}^{d} \Bigg( &\mu_b(Z) q_a(A=b \mid X) + \left[ \frac{\one(A=b)}{\pi_b(X)}\left\{ Y - \mu_b(X) \right\} \right] q_a(A=b \mid X) + \mu_b(X) \varphi_{q_a}(Z; b) \Bigg) \label{eq:eif}
    \end{align}
    where 
    \begin{align}
        \varphi_{q_a}(Z; b) &= \varphi_S(Z) \{ \one(b=a) - \pi_b(X) \} + \Big\{ 1 -S(X \in C_X) \Big\} \Big\{ \one(A=b) - \pi_b(X) \Big\} \text{, and } \label{eq:if-rho_a} \\
        \varphi_{S}(Z) &= \sum_{b=1}^{d} \bigg( s^\prime\{ \pi_b(X) \} \left\{ \one(A=b) - \pi_b(X) \right\} \bigg) \prod_{c \neq b}^{d} s\{ \pi_c(X) \}, \nonumber
    \end{align}
    and $S(X \in C_X)$ is defined in \eqref{eq:smooth-trim}.
\end{theorem}

The un-centered EIF in \eqref{eq:eif} is more complex than typical EIFs for standard estimands.  However, at a high level, it takes the usual form of a plug-in plus the sum of weighted residuals, as $\varphi_{q_a}(Z; b)$ and $\varphi_S(Z)$ each consist of sums of weighted residuals.  It is possible to construct a one-step estimator for $\psi_a$ with the sample average of the un-centered EIF. Adding the weighted residuals to the plug-in debiases the plug-in estimate. The EIF could also be used to construct other efficient estimators, such as a targeted maximum likelihood estimator \citep{van2006targeted}. The one-step estimator we consider is also referred to as a double machine learning estimator \citep{chernozhukov2018double}. We focus on the one-step estimator because it has the same asymptotic guarantees as other estimators and is arguably simpler to construct.  The next result shows the one-step estimator's desirable asymptotic properties: its bias is doubly robust-style, in the sense that it is upper bounded by the sum of the products of errors from the nuisance function estimators.

\begin{theorem} \label{thm:dr-est} 
    \textbf{(Second order bias)} Let $\psi_a$ denote the identified parameter in Proposition~\ref{prop:id} where $q_a$ is defined in \eqref{eq:smooth-prop-scores} and $V = X$. Suppose access to nuisance function estimates $\widehat \pi_b(X)$ and $\widehat \mu_b(X)$ for $b \in \{1,\dots, d\}$ independent from the observed sample, and construct a one-step estimator for $\psi_a$ as $\widehat \psi_a = \bbP_n \left\{ \widehat{\varphi}_{a}(Z) \right\}$, where  $\varphi_{a}(Z)$ is defined in \eqref{eq:eif}. Under the conditions of Theorem~\ref{thm:eif}, suppose further that $\left| \widehat \mu_b(X) \right|$ is uniformly bounded for all $b \in \{1, \dots, d \}$. Then, 
    \begin{align}
        &\left| \bbE \left( \widehat \psi_a - \psi_a \right) \right| \lesssim \sum_{b=1}^{d} \left|  \bbE \left[ \left\{ \widehat \pi_b(X) - \pi_b(X) \right\} \left\{ \widehat \mu_b(X) - \mu_b(X) \right\}  \frac{\widehat q_a(A=b \mid X)}{\widehat \pi_b(X)} \right] \right| \nonumber \\
        &+ \sum_{b=1}^{d} \left( \left\lVert  \widehat \mu_b - \mu_b \right\rVert + \left\lVert \widehat \pi_b - \pi_b \right\rVert \right) \left\{ \sum_{c=1}^d \left\lVert s^\prime(\pi_c) ( \widehat \pi_c - \pi_c ) \right\rVert  \right\} \nonumber \\
        &+ d \Bigg\{ \sum_{b=1}^d \left\lVert s^{\prime \prime} ( \pi_b )^{1/2} ( \widehat \pi_b - \pi_b) \right\rVert^2 + \sum_{b=1}^{d} \sum_{c < b} \left\lVert s^\prime (\pi_b) (\widehat \pi_b - \pi_b ) \right\rVert \left\lVert s^\prime (\pi_c) (\widehat \pi_c - \pi_c ) \right\rVert  \Bigg\} \label{eq:dr-est-smooth-trim} 
    \end{align}
\end{theorem}

This result shows that the one-step estimator has a second order bias consisting of sums of products of errors in estimating the propensity score and outcome regression. We first describe the conditions required for this result. The result assumes access to independent nuisance estimates. While this assumption is not strictly necessary for analyzing the bias of the estimator, it is essential for establishing its limiting distribution in Corollary~\ref{cor:normal}, as it allows us to avoid imposing Donsker or other complexity conditions \citep{van1996weak, chen2022debiased}. Moreover, this assumption is mild and can be ensured through cross-fitting, which typically involves randomly splitting the data into $k$ folds (five and ten are common), training the nuisance function estimators on $k-1$ folds, and evaluating the one-step estimator on the $k^{th}$ held-out fold \citep{robins2008higher, zheng2010asymptotic}. Full-sample efficiency can be retained by rotating the folds. The remaining condition is a mild boundedness condition, that the regression estimates are uniformly bounded.
    
\medskip

Next, we describe the bias term itself, in \eqref{eq:dr-est-smooth-trim}. It consists of three pieces.  The first summand arises from estimating $\psi_a$ if the interventional propensity scores $q_a(A=b \mid X)$ were known, and is the canonical product of errors in estimating the nuisance functions. It could be upper bounded by a product of root-mean-squared-errors using H\"{o}lder's inequality and the Cauchy-Schwartz inequality, supposing $\left| \frac{\widehat q_a(A=b \mid X)}{\widehat \pi_b(X)} \right|$ were uniformly bounded for all $b \in \{1, \dots, d\}$.  We leave it as it is to highlight how different interventions could affect the error bound through $\widehat q_a(A=b \mid X)$. The second summand is the product of residuals in estimating $q_a(A=b \mid X)$ and $\mu_b(X)$, while the third summand is the product of residuals that arises from the debiased estimator of $S(X \in C_X)$ within $q_a(A=b \mid X)$. This third term shows a strong dependence on estimating the propensity scores.  This is because estimating $S(X \in C_X)$ and its EIF involves estimating all the propensity scores simultaneously.  Therefore, the double sum over $b$ and $c$ arises, as does an outer factor of $d$.  This is important because it shows how estimating these parameters depends on the dimension of the treatment: as the number of possible treatments increases, the bias convergence rate slows.

\begin{remark}
    In this paper, we assume that the smoothing functions $s(\cdot)$ are fixed with sample size. This allows for straightforward derivation of the bias bound above and the limiting distribution guarantee in the next result. Moreover, it agrees with a typical fixed-sample data analysis, where one chooses a fixed $s(\cdot)$. This is the approach we take in our data analysis. However, one could investigate convergence guarantees when allowing the smoothing function to change with sample size to minimize smooth approximation error, if the target of interest is the trimmed treatment-specific means. That is beyond the scope of this work, but has been considered previously in causal inference (e.g., \citet{levis2024nonparametric}).
\end{remark}

Theorem~\ref{thm:dr-est} establishes that a one-step estimator has second order bias. As a consequence, inference is possible when the nuisance functions satisfy nonparametric convergence rates, as in the following result.

\begin{corollary} \label{cor:normal}
    \textbf{(Normal limiting distribution)} Construct one-step estimates $\{ \widehat \psi_a \}_{a=1}^{d}$ according to Theorem~\ref{thm:dr-est} and suppose the conditions of Theorem~\ref{thm:dr-est} hold for all $a \in \{1, \dots, d\}$.  In addition, suppose Assumption~\ref{asmp:intermediate-positivity} holds,
    \begin{equation} \label{eq:ifs-consistent}
        \sum_{a=1}^{d} \bbE \{ \widehat \varphi_{a}(Z) - \varphi_{a}(Z) \} = o_{\bbP}(1),
    \end{equation}
    and
    \begin{equation} \label{eq:bias-convergence}
        \sum_{a=1}^{d} \sum_{b=1}^{d} \lVert \widehat \mu_a - \mu_a \rVert  \lVert \widehat \pi_b - \pi_b \rVert + d \sum_{a=1}^{d} \sum_{b \leq a} \lVert \widehat \pi_a - \pi_a \rVert \lVert \widehat \pi_b - \pi_b \rVert = o_{\bbP}(n^{-1/2}).
    \end{equation}
    Then, 
    $$
    \sqrt{n} \begin{pmatrix} \widehat \psi_1 - \psi_1 \\ \vdots \\ \widehat \psi_ d - \psi_d \end{pmatrix} \indist N(0, \Sigma)
    $$
    where $e_i^T \Sigma e_j = \cov \{ \varphi_{i}(Z), \varphi_{j}(Z) \}$ and $e_i$ is the $i^{th}$ standard basis vector in $\bbR^d$.
\end{corollary}

This result establishes when the vector of one-step estimates $\{ \widehat{\psi}_a \}_{a=1}^{d}$ achieves a normal limiting distribution. The result depends on the conditions outlined in Theorem~\ref{thm:dr-est}. Moreover, it requires intermediate positivity, which rules out propensity scores arbitrarily close to zero, and requires the mild assumption of consistency of the estimated EIFs, in \eqref{eq:ifs-consistent}. Crucially, it also demands that the product of errors in the nuisance function estimators converges to zero at a rate of \(n^{-1/2}\), in \eqref{eq:bias-convergence}. This requirement is attainable under nonparametric conditions on the nuisance functions---such as smoothness, sparsity, or bounded variation---when \(n^{-1/4}\) convergence rates can be achieved for each nuisance function \citep{gyorfi2002distribution}.  

\medskip

When Corollary~\ref{cor:normal} applies, it is possible to estimate and conduct inference on a variety of comparisons. One could construct a joint confidence interval for rankings of causal parameters (see, \citet{klein2020joint}). However, when there are many parameters, a joint confidence interval might be very wide, making it preferable to focus on other parameters. Instead, it is possible to estimate the $q$-quantile of the ranking and construct a confidence interval for that $q$-quantile. Alternatively, recent work has proposed methods for estimating the ``$\tau$-best'' parameters, which could, for example, be straightforwardly adapted to estimate the $\tau$-best providers in provider profiling \citep{mogstad2024inference}. In provider profiling, researchers typically focus on comparing each parameter to some benchmark, like the median or mean parameter. One could estimate and conduct inference on such a comparison using Corollary~\ref{cor:normal} and the delta method.

\section{Data analysis: provider profiling} \label{sec:data-analysis}

We apply our method to analyze Medicare inpatient claims for end-stage renal disease (ESRD) beneficiaries undergoing kidney dialysis in 2020. Our analysis is publicly available at 
\url{www.github.com/alecmcclean/comparisons-positivity}.

% This analysis is an example of the general challenge of provider profiling \citep{welch1994physician, auerbach1999principles}. Provider profiling aims to quantify the performance of healthcare providers—clinicians, hospitals, and others—using standardized quality metrics, such as 30-day readmission rates and post-discharge mortality. In the U.S., provider profiling is seen as a key tool for assessing healthcare quality and encouraging improvements \citep{goldfield2003profiling}. The Centers for Medicare and Medicaid Services (CMS) implement the ESRD Quality Incentive Program, which penalizes dialysis facilities with poor performance, such as higher-than-expected readmission rates, through payment reductions \citep{cmsesrdqip}.

\subsection{Data} 

The data includes claims information from the United States Renal Data System (USRDS) \citep{usrds2022adr}. The outcome is all-cause unplanned hospital readmission within 30 days of discharge. The dataset contains patient demographic information (sex, age at first ESRD service, race, and ethnicity), physical attributes (body mass index and parameter status), social factors (substance/alcohol/tobacco use and employment status), and clinical characteristics (length of hospital stay, time since ESRD diagnosis, and dialysis mode). It also captures the cause of ESRD (diabetes, hypertension, primary glomerulonephritis, or other) and prevalent comorbidities (in-hospital COVID-19, heart failure, coronary artery disease, cerebrovascular accident, peripheral vascular disease, cancer, and chronic obstructive pulmonary disease).

\medskip

Although in reality a patient can have multiple claims, to simplify the illustration we assume each claim corresponds to a unique patient. Moreover, we focus on claims from the ten most common providers in New York State. This reduction resulted in a dataset of 11,052 claims, distributed roughly evenly across ten providers. Despite this simplification, positivity violations persisted for several providers. Indeed, there were $31$ patient-provider combinations with an estimated propensity score of zero, and $10.0$\% of patient-provider combinations had an estimated propensity score below $0.01$. Given these positivity violations, the dataset remains suitable for demonstrating our methods: we consider smooth trimmed treatment-specific means that facilitate comparability of provider efficacy while simultaneously addressing positivity violations. \color{black}

\subsection{Methods}

We analyzed the data using the estimator detailed in Section~\ref{sec:eif-estimator} for smooth trimmed treatment-specific means. Specifically, we designed estimators that satisfy the comparability criterion for $V = X$ with smooth approximation $s(x) = 1 - \exp(-100x)$.  To construct the one-step estimator, we employed two-fold cross-fitting and estimated the generalized propensity score and outcome regression with ensemble (or stacking \citep{breiman1996stacked}) estimators via the \texttt{SuperLearner} package \citep{polley2024super, van2007super} in \texttt{R} \citep{r2024language}. The ensemble included the mean, a random forest with default parameter settings \citep{wright2017fast}, a generalized additive model with no interactions \citep{hastie2024gam}, and generalized linear models with no interaction terms and with all interactions terms. 

\subsection{Results}

Figure~\ref{fig:results} presents the main results. The x-axis indicates the targeted provider. For example, the left-most point shows the effect of a smooth trimmed treatment-specific intervention targeting provider \texttt{I} (the provider IDs are anonymized). The black point and whiskers represent the point estimate and 95\% pointwise confidence interval. The red horizontal line indicates the observed average with no intervention.

\medskip

\begin{figure}[ht]
    \centering
    \includegraphics[height=4in]{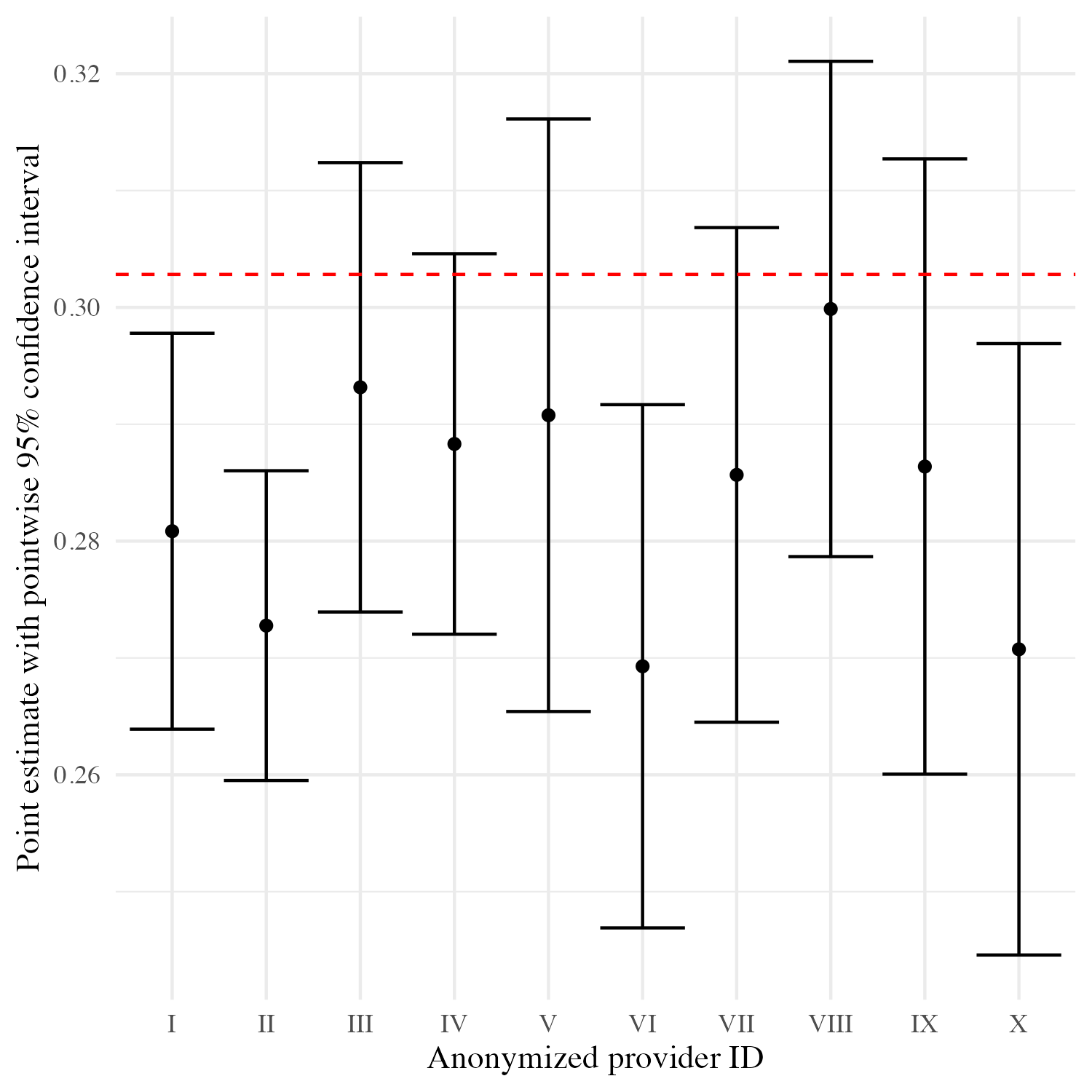}
    \caption{Results for different sets of interventions}
    \label{fig:results}
\end{figure}

Figure~\ref{fig:results} shows minimal variation in provider performance once statistical uncertainty is accounted for, as indicated by the overlapping pointwise confidence intervals, although there is some variation in the point estimates: we estimate provider \texttt{VI} has the lowest counterfactual 30-day readmission rate, at 0.269, while provider \texttt{VIII} has the highest counterfactual readmission estimate, at 0.300. To better understand whether provider \texttt{VIII} was markedly worse than other providers, we also estimated and constructed confidence intervals for the difference between its performance and the performance of the other providers. Those results are presented in Table~\ref{tab:differences}. The table shows two 95\% confidence intervals do not include zero, for providers \texttt{II} and \texttt{VI}.  This indicates there is a statistically significant difference between the readmission rate at these providers and the readmission rate at provider \texttt{VIII}. These results exhibit narrower confidence intervals than we might have anticipated from Figure~\ref{fig:results} because there is a high positive correlation between estimates for each provider, and therefore the standard error of the estimator for the difference in readmission rates can be smaller than the standard errors for estimating each readmission rate separately.

\begin{table}
    \centering
    \begin{tabular}{>{\raggedright}p{2cm} >{\raggedleft\arraybackslash}p{4.5cm}>{\raggedleft\arraybackslash}p{3.5cm}}  % Adjust column width as needed
        \toprule
        \textbf{Anonymized provider identifier} & \textbf{Difference between 30-day readmission rate for provider \texttt{VIII} and this provider} & \textbf{95\% confidence interval} \\
        \midrule
        \texttt{I}      & 0.019 & [-0.006, 0.044]   \\
        \texttt{II}     & 0.027 & [0.004, 0.050]    \\
        \texttt{III}    & 0.007 & [-0.020, 0.033]   \\ 
        \texttt{IV}     & 0.012 & [-0.013, 0.036]   \\
        \texttt{V}      & 0.009 & [-0.022, 0.040]   \\
        \texttt{VI}     & 0.031 & [0.002, 0.060]    \\
        \texttt{VII}    & 0.014 & [-0.014, 0.043]   \\
        \texttt{IX}     & 0.013 & [-0.019, 0.046]   \\
        \texttt{X}      & 0.029 & [-0.003, 0.061]   \\ 
        \bottomrule
    \end{tabular}
    \caption{Difference between 30-day readmission rate for provider \texttt{VIII} and readmission rate for other top ten providers in New York State, with 95\% confidence interval}
    \label{tab:differences}
\end{table}

\medskip

Qualitatively, these results suggest that the top ten largest dialysis facilities in New York State exhibited similar performance in terms of 30-day unplanned readmission rates during the period of study, but that the worst performing provider --- provider \texttt{VIII} --- performed statistically significantly worse than two of the other providers --- providers \texttt{II} and \texttt{VI}.  This information could be used to inform policy recommendations, and help provider \texttt{VIII} target improvements that close the performance gap between it and other providers. From a methodological perspective, this analysis illustrates how to construct sets of estimates that yield useful comparisons and remain identifiable under positivity violations. The presented estimates satisfy the comparability criterion for $V = X$ and for identification only require a mild positivity assumption that there exists some set of patients that had a non-zero probability of attending every provider.

\section{Simulations} \label{sec:simulations}

We demonstrate the method through simulations estimating the smooth trimmed treatment effect with binary treatment. We implement the one-step estimator from Section~\ref{sec:eif-estimator} for the two smooth trimmed treatment-specific means:
\begin{align*}
    \psi_0 &= \bbE \left( \big[ S(X) + \{ 1 - S(X) \} \pi_0(X) \big] \mu_0(X) + \{ 1 - S(X) \} \pi_1(X) \mu_1(X) \right) \text{ and } \\
    \psi_1 &= \bbE \left( \{ 1 - S(X) \} \pi_0(X)  \mu_0(X) + \big[ S(X) + \{ 1 - S(X) \} \pi_1(X) \big] \mu_1(X) \right),
\end{align*}
where $S(X)$ approximates the trimming indicator $\one \{ 0 < \pi_1(X) < 1 \}$ with a smooth trimming indicator; $S(X) = s\{ \pi_1(X) \} s\{ 1 - \pi_1(X) \}$ where $s(x) = 1- \exp(-20 x)$. The target parameter is their difference, which equals a smooth trimmed treatment effect: $\psi_1 - \psi_0 = \bbE \left[ \{ \mu_1(X) - \mu_0(X) \} S(X) \right]$.

\medskip

The data generating process is illustrated by Figure~\ref{fig:all_plots}. There is a single covariate $X \sim \text{Unif}(0,1)$. The propensity score has positivity violations so that $\pi_1(X) = 0$ when $X \leq 0.1$ and $\pi_1(X) = 1$ when $X \geq 0.9$, which is shown in the top left plot. The conditional average treatment effect is $\bbE(Y^1 - Y^0 \mid X) = X^2$, in the top right plot. The trimming indicator $\one\{ 0 < \pi_1(X) < 1\}$ and its smooth approximation $S(X)$ are shown in the bottom plot. We used sample sizes of 100, 1,000, 10,000, and 100,000, and for each data generating process and sample size, we constructed 200 datasets.

\medskip

To evaluate robustness to nuisance estimation error, we simulated nuisance estimation error by adding random noise to the true nuisances --- the outcome model and propensity score --- which allowed us to control the convergence rates of the nuisance estimators as sample size increased.

\medskip

Figure~\ref{fig:coverage} presents coverage results for 95\% Wald-type confidence intervals for $\psi_1 - \psi_0 = \bbE \left[ \{ \mu_1(X) - \mu_0(X) \} S(X) \right]$ across different nuisance convergence rate scenarios. The estimator achieves nominal coverage when the product of nuisance estimation errors converges faster than $n^{-1/2}$ (as was predicted by Corollary~\ref{cor:normal}), but coverage deteriorates when this condition is violated (top left panel). This confirms the necessity of the convergence rate conditions for valid inference.

\begin{figure}[ht]
    \centering
    \includegraphics[width=\textwidth]{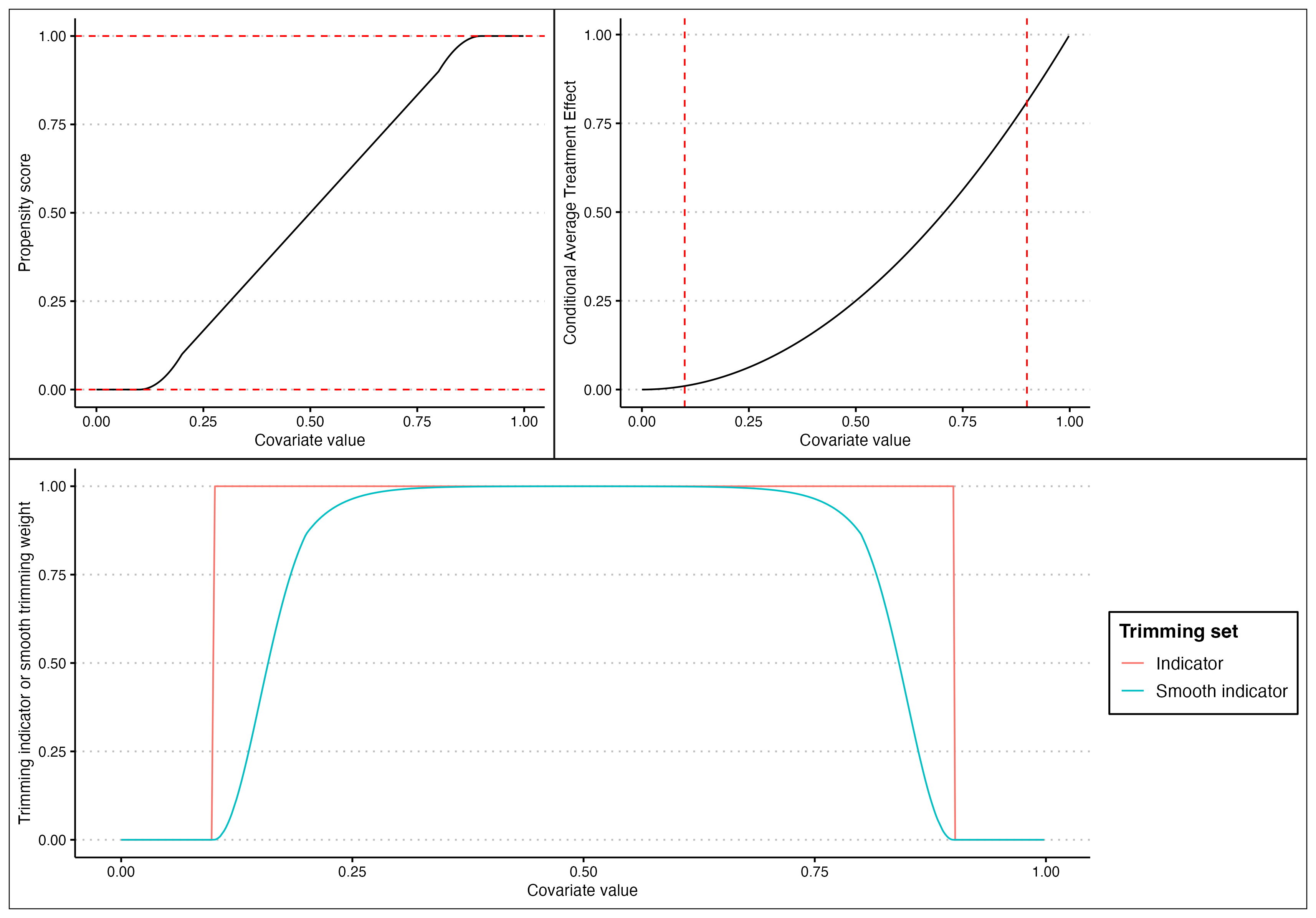}
    \caption{Data generating process illustration showing propensity score with positivity violations (top left), conditional treatment effect (top right), and smooth trimming function (bottom).}
    \label{fig:all_plots}
\end{figure}

\begin{figure}[ht]
    \centering
    \includegraphics[width=5in]{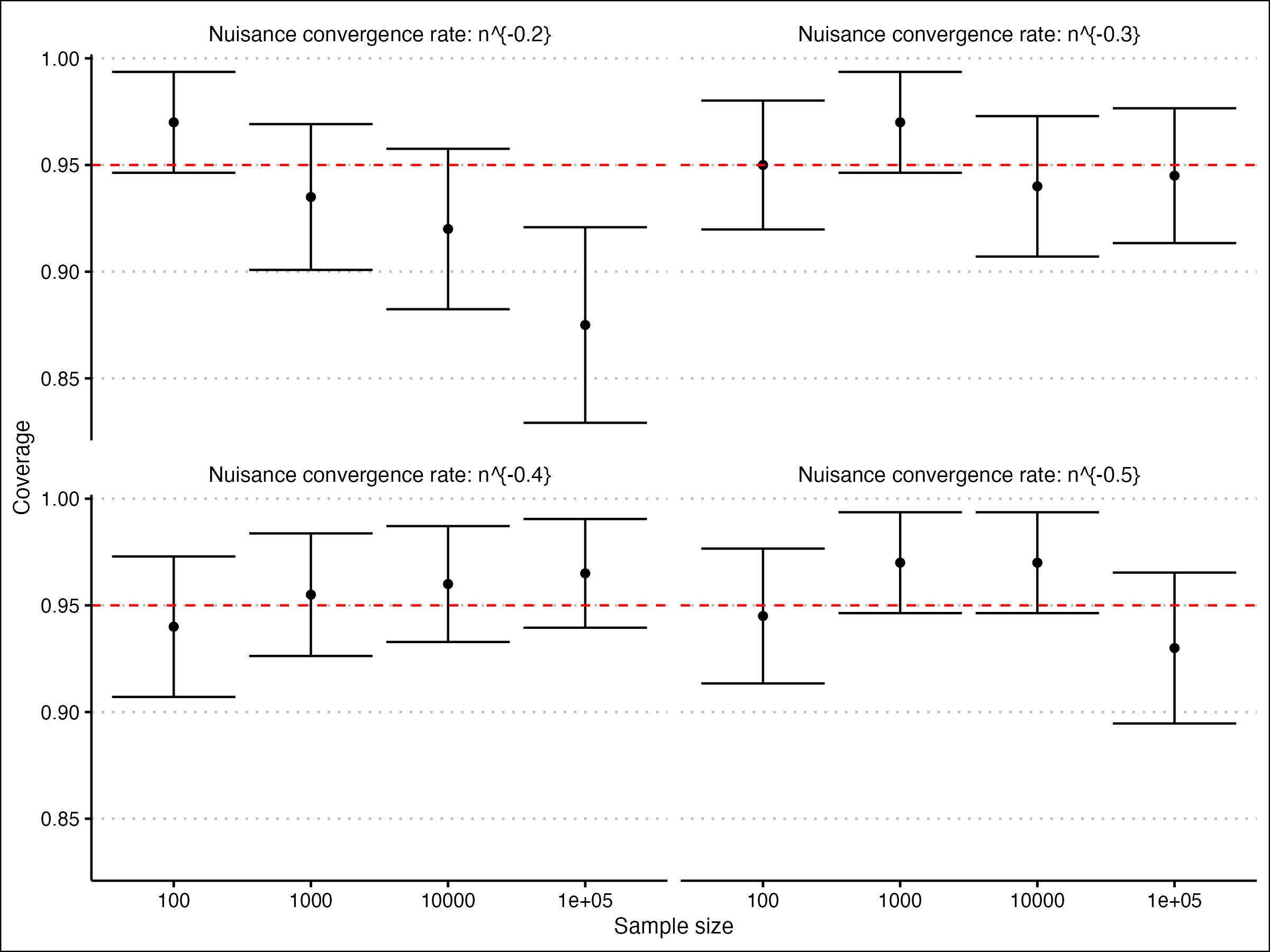}
    \caption{Coverage of 95\% confidence intervals for the smooth trimmed treatment effect across different nuisance function convergence rate scenarios and sample sizes. Error bars are 95\% Wald-style confidence intervals for the estimated coverage.}
    \label{fig:coverage}
\end{figure}

\section{Discussion} \label{sec:discussion} 

In this paper, we introduced new methods for comparing causal parameters in the presence of positivity violations with multi-valued treatments. We proposed a simple comparability criterion, stipulating that if one treatment's conditional treatment-specific mean is larger than another's, then the causal parameter targeting the first treatment is larger (and if the conditional treatment-specific means are equal, the parameters are equal). We then showed that many common examples fail to satisfy this property but established that parameters satisfying these properties can be identified under a mild positivity assumption. We proposed several examples, including trimmed and smooth trimmed effects, that satisfy the comparability criterion. We then developed doubly robust-style estimators for smooth trimmed effects, which achieve parametric convergence rates and normal limiting distributions, even with nonparametric nuisance function estimators. Our approach extends the applicability of causal inference methods to settings with positivity violations, such as large-scale healthcare provider profiling \citep{wu2022improving}. Finally, we illustrated the utility of these methods through a simulation study and an analysis of dialysis facility performance in New York State.

\bigskip

There are many topics for future study. A natural extension of our work could consider data generating processes where the dimension of treatment is very large or can grow with sample size (which occurs in provider profiling, \citep{nguyen2023high, he2013evaluating, varewyck2014shrinkage}), or where multiple observations occur for the same subject over time.  Meanwhile, to understand heterogeneity, it is important to construct parameters that can vary with covariate information, so this could be a complementary avenue of investigation.

\section*{Acknowledgments}

The project was partially supported by the Alzheimer's Association (AARG-23-1077773), National Institute on Aging (K02AG076883), and the Department of Population Health and Center for the Study of Asian American Health at the NYU Grossman School of Medicine (U54MD000538). The data reported here have been supplied by the United States Renal Data System (USRDS). The interpretation and reporting of these data are the responsibility of the author(s) and in no way should be seen as an official policy or interpretation of the U.S. government.

\bigskip

\noindent The authors thank Zach Branson, Edward Kennedy, the NYU causal inference group, and the Columbia causal inference learning group for helpful feedback.

\section*{Data Availability Statement}

The datasets used and/or analyzed during the current study are available from the United States Renal Data System (USRDS) upon data use agreement approval. Per the Data Use Agreement (DUA) between the authors and USRDS, the release of the data or the deposition of data into publicly available repositories or to individuals is not allowed.

\section*{References}

\vspace{-0.3in} 
\bibliographystyle{plainnat}
\bibliography{references}

\appendix

\vspace{0.5in}

\noindent {\LARGE \textbf{Appendix}}

\bigskip
\noindent The appendix is organized as follows:
\begin{itemize}
    \item Appendix~\ref{app:comp-extras} provides additional information about the results in Section~\ref{sec:conditions} and the properties needed to satisfy $V$-comparability and identifiability.
    \item Appendix~\ref{app:examples} provides the general class of examples, which include trimmed and smooth trimmed treatment-specific means as a special case.
    \item Appendix~\ref{id-est:extra} develops general results for identification and estimation. 
    \item Appendix~\ref{app:comparability-proofs} contains proofs of the results in Section~\ref{sec:conditions}.
    \item Appendix~\ref{app:eif-estimator} contains proofs of the results in Section~\ref{sec:eif-estimator}.
\end{itemize}

\section{Comparability and positivity} \label{app:comp-extras}

\subsection{Two interventions that are equivalent to comparability}

First, we establish that satisfying two intuitive properties is equivalent to satisfying $V$-comparability.  These properties can be used to establish what examples can satisfy comparability and what positivity assumption is necessary to identify them.

\begin{property} \label{property:increase-a-conditional}
    Let $\{ q_a(A=b \mid V) \}_{b=1}^{d}$ denote the interventional propensity scores of a dynamic stochastic intervention which varies with covariates $V \subseteq X$ and targets treatment $a$.  Then, $\bbP \{ q_a(A = a \mid V) \geq \pi_a(V) \} = 1$ and $\bbP \{ q_a (A = a \mid V) > \pi_a(V) \} > 0$, and $\bbP \{ q_a(A = b \mid V) \leq \pi_b(V) \} = 1$ for all $b \neq a$.
\end{property}

\begin{property} \label{property:non-ab-equal}
    Let $\{ q_a(A=c \mid V) \}_{c=1}^{d}$ and $\{ q_b(A=c \mid V)\}_{c=1}^{d}$ denote the interventional propensity scores of two dynamic stochastic interventions which vary with $V \subseteq X$ and target treatments $a$ and $b$, respectively. Then, $\bbP \{ q_a(A = c \mid V) = q_b(A = c \mid V) \} = 1$  for all $c \notin \{a, b\}.$
\end{property}

Property~\ref{property:increase-a-conditional} states that an intervention targeting treatment $a$ satisfies the following: (1) the probability of receiving treatment $a$ either increases or remains the same, (2) the probability of receiving treatment $a$ increases for a non-zero-probability subset of subjects, and (3) the probability of receiving any treatment besides $a$ either decreases or remains the same. Property~\ref{property:non-ab-equal} is defined for two interventions targeting treatments $a$ and $b$. It asserts that the interventional propensity scores are almost surely equal for all non-target treatments. The next result establishes that a set of causal parameters satisfies $V$-comparability if and only if it satisfies properties~\ref{property:increase-a-conditional} and~\ref{property:non-ab-equal}.

\begin{lemma} \label{lem:dynamic}
    \textbf{(Equivalence between $V$-comparability and properties~\ref{property:increase-a-conditional} and \ref{property:non-ab-equal})} Let $\big\{ \psi_a \big\}_{a=1}^{d} = \big\{ \bbE(Y^{Q_a}) \big\}_{a=1}^{d}$ denote a set of parameters defined by dynamic stochastic interventions that vary with covariates $V \subseteq X$ and target treatments $1, \dots, d$, respectively. The set satisfies criterion~\ref{crit:comp} for all $(a,b) \in \{1, \dots, d \} \times \{ 1, \dots, d\}$ if and only if it satisfies property~\ref{property:increase-a-conditional} separately for all $a \in \{1, \dots, d\}$ and property~\ref{property:non-ab-equal} for all $(a,b) \in \{1, \dots, d \} \times \{ 1, \dots, d\}$.
\end{lemma}

The proof for Lemma~\ref{lem:dynamic} appears in Appendix~\ref{app:comparability-proofs}. Lemma~\ref{lem:dynamic} is important for two reasons. First, it is easier to establish the necessary positivity assumption for identification from properties~\ref{property:increase-a-conditional} and \ref{property:non-ab-equal} than from criterion~\ref{crit:comp}. Second, the two properties suggest how to construct interventions.  In Appendix~\ref{app:examples}, we use them to generate intuition to construct interventions satisfying criterion~\ref{crit:comp} and identifiability simultaneously.

\smallskip

Finally, for establishing Theorem~\ref{thm:positivity}, we introduce a third property, which we term `$q$-weak positivity.'
\begin{property} \label{property:weak-positivity}
    \textbf{($q$-weak positivity)} Let $\{ q(A=a \mid V) \}_{a=1}^{d}$ denote the interventional propensity scores of a dynamic stochastic intervention that varies with covariates $V \subseteq X$ and let $W = X \setminus V$. The intervention satisfies $q$-weak positivity if $\bbP \big\{ \pi_a(W, V) = 0 \implies q(A=a \mid V) = 0 \big\} = 1$ for all $a \in \{1, \dots, d\}$; i.e., almost surely, if $\pi_a(W, V)$ equals zero then so does $q(A=a \mid V)$.
\end{property}

This property ensures $\bbE(Y \mid A=a, X)$ exists for all values $V \subseteq X$ where $q(A = a \mid V) > 0$.  It is a typical assumption for identifying counterfactual parameters based on dynamic stochastic interventions with observational data \citep{kennedy2019nonparametric}. It is a weaker condition than weak positivity in Assumption~\ref{asmp:weak-positivity}. In fact, it is not an assumption about the underlying propensity scores $\pi_a(X)$; instead, property~\ref{property:weak-positivity} can be achieved through careful design of the interventional propensity scores $q(A=a \mid V)$. To establish Theorem~\ref{thm:positivity}, we show that property~\ref{property:weak-positivity} can be satisfied alongside properties~\ref{property:increase-a-conditional} and \ref{property:non-ab-equal}.

\color{black}
\subsection{Weak positivity and comparability}

Here, we examine the tension between satisfying $q$-weak positivity (property~\ref{property:weak-positivity}) and properties~\ref{property:increase-a-conditional} and \ref{property:non-ab-equal}.  The results here will be used to prove Theorem~\ref{thm:positivity}. We show that satisfying properties~\ref{property:increase-a-conditional}-\ref{property:weak-positivity} requires a positivity assumption that there is a positive probability set which can receive all treatments.  The first result shows that a very mild positivity assumption is required for properties~\ref{property:increase-a-conditional} and \ref{property:weak-positivity} to hold.
\begin{proposition} \label{prop:positivity}
    Let $\big\{ q_a(A=b \mid V) \big\}_{b=1}^{d}$ denote the propensity score of dynamic stochastic intervention varying over covariates $V \subseteq X$. If $q_a$ satisfies property~\ref{property:increase-a-conditional} it must be the case that $\bbP(A = a) < 1$.  Meanwhile, if $q_a$ satisfies properties~\ref{property:increase-a-conditional} and \ref{property:weak-positivity}, then it must be the case that $\bbP(A = a) > 0$.
\end{proposition}

Proposition~\ref{prop:positivity} rules out an extreme violation of positivity --- when $\bbP(A = a) = 1$ or $\bbP(A = a) = 0$.  Applying this to all interventions $\{ q_1, \dots, q_d \}$ strengthens this assumption somewhat and rules out $\bbP(A = a) \in \{0, 1\}$ for all $a \in \{1, \dots, d\}$. However, notice that this is a much milder assumption than strong positivity in Assumption~\ref{asmp:strong-positivity}, because it only requires $\pi_a(V)$ is bounded away from zero and one for a positive probability set, not all $V$. 

\bigskip

The next result shows what positivity condition is necessary to satisfy properties~\ref{property:increase-a-conditional}-\ref{property:weak-positivity} simultaneously. 

\begin{lemma} \label{lem:positivity}
     Let $\big\{ \psi_a \big\}_{a=1}^{d} = \big\{ \bbE(Y^{Q_a}) \big\}_{a=1}^{d}$ denote a set of parameters defined by dynamic stochastic interventions that vary with covariates $V \subseteq X$ and target treatments $1, \dots, d$, respectively.  If $\left\{ q_a \right\}_{a=1}^{d}$ satisfy $q$-weak positivity (property~\ref{property:weak-positivity}) and property~\ref{property:increase-a-conditional} separately and satisfy property~\ref{property:non-ab-equal} for all pairs $\{q_a, q_b\}_{(a,b) \in \{1, \dots, d \} \times \{1, \dots, d\}}$ together, it must be the case that
    \begin{equation} 
        \bbE \left( \prod_{a=1}^{d} \one \Big[ \bbP \left\{ \pi_a(X) > 0 \mid V \right\} = 1 \Big]  \right) > 0.
    \end{equation}
    In other words, \eqref{eq:positivity-condition} is necessary for properties~\ref{property:increase-a-conditional}-\ref{property:weak-positivity} to hold simultaneously. 
\end{lemma}
\noindent Importantly, notice that, when combined with Lemma~\ref{lem:dynamic}, Lemma~\ref{lem:positivity} implies Theorem~\ref{thm:positivity}.

\section{General examples} \label{app:examples}

In this section, we propose general examples that satisfy $V$-comparability and are identifiable under only the minimum positivity assumption from Theorem~\ref{thm:positivity}. Indeed, properties~\ref{property:increase-a-conditional} and \ref{property:non-ab-equal} and the necessary positivity assumption in \eqref{eq:positivity-condition} suggest a strategy for constructing examples that are comparable and identifiable. According to property~\ref{property:increase-a-conditional}, an intervention targeting treatment $a$ should increase the probability of receiving treatment $a$ and decrease the probability of receiving other treatments. According to property~\ref{property:non-ab-equal}, two interventions targeting treatments $a$ and $b$ should have the same interventional propensity score at all other treatments $c \notin \{a, b\}$. Finally, the positivity assumption suggests constructing interventions that only target the subset of subjects with a non-zero probability of receiving every treatment (the subjects in the trimmed set $C_V$, defined in \eqref{eq:trimmed-set} in the main paper). With this intuition, we can construct a general shift intervention that satisfies properties~\ref{property:increase-a-conditional} and \ref{property:non-ab-equal} and yields an identifiable parameter under \eqref{eq:positivity-condition} in the main paper.

\begin{example} \emph{\textbf{General shift targeting treatment $a$.}} \label{example:general}
    \begin{equation*}
        q_a(A=b \mid V = v) = \one(v \in C_V) \rho_a(A=b \mid V=v) + \one(v \notin C_V) \pi_b(v)
    \end{equation*}
    where
    \begin{equation} \label{eq:def-rho}
        \rho_a(A=b \mid V=v) = \one(b \neq a) f \{ \pi_b(v) \} + \one(b = a) \left[ 1 - \sum_{b \neq a} f\{ \pi_b(v) \} \right] 
    \end{equation}
    and \( f: [0,1) \to [0,1) \) satisfies \( f(x) < x \). 
\end{example}
The target is treatment $a$. The indicator function $\one(v \in C_V)$ ensures that only subjects who could receive every treatment are intervened upon, which allows the resulting parameters to be identifiable under the positivity assumption in \eqref{eq:positivity-condition} in the main paper.  Meanwhile, in $\rho_a$, the intervention is defined explicitly for treatments $b \neq a$ but only implicitly for $a$.  It decreases the probability of receiving treatment $b \neq a$, and implicitly increases the likelihood of receiving treatment $a$ as one minus the sum of the non-target propensity scores.  This allows the resulting parameters to satisfy properties~\ref{property:increase-a-conditional} and \ref{property:non-ab-equal} and, by extension, $V$-comparability. This construction can be applied to trimmed treatment-specific means or stochastic interventions, like multiplicative shifts and exponential tilts, or any general function satisfying $f(x) < x$.
\begin{example} \label{example:tsm} \emph{\textbf{Treatment specific mean.}} $f(x) = 0$.
\end{example}

\begin{example} \label{example:multiplicative}
    \emph{\textbf{Multiplicative shift.}}  $f(x) = \delta x$ for $\delta \in [0,1)$. 
\end{example}

\begin{example} \label{example:tilt}
    \emph{\textbf{Exponential tilt.}} $f(x) = \frac{\delta x}{\delta x + 1 - x}$ for $\delta \in [0,1)$.
\end{example}

As with smooth trimming in the main paper, we can replace the non-smooth indicator $\one(V \in C_V)$ by a smooth approximation.
\begin{example} \emph{\textbf{Smooth general shift targeting treatment $a$.}} \label{example:smooth-general}
\begin{align}
    q_a(A=b \mid V) &= S(V \in C_V) \rho_a(A=b \mid V) + \left\{ 1 -S(V \in C_V) \right\} \pi_b(V) \label{eq:q-def}
\end{align}
where $f: [0,1) \to [0,1)$ satisfies $f(x) < x$ and $\rho_a$ and $S$ are defined in \eqref{eq:def-rho} in this section and \eqref{eq:smooth-trim} in the main paper.
\end{example}

We conclude this section with two remarks.
\begin{remark}
    Typically, stochastic interventions such as multiplicative shifts and exponential tilts are employed to address positivity violations \citep{kennedy2019nonparametric, wen2023intervention}. Here, focusing on the trimmed set of subjects ensures positivity across that set. Nevertheless, stochastic interventions may still be preferable to trimmed treatment-specific means, for two reasons. First, under near violations of the positivity assumption in \eqref{eq:positivity-condition} in the main paper, the nonparametric efficiency bound for estimating the trimmed treatment-specific means may increase dramatically, similar to how the bound increases for standard treatment-specific means under near violations of weak positivity. Stochastic interventions can ameliorate this issue. Second, stochastic interventions may correspond to a more relevant parameter for policy purposes. The trimmed treatment-specific means consider the scenario where every subject in the trimmed set receives a treatment, which might not be practically feasible. By contrast, a stochastic intervention can represent milder shifts in the probability of receiving treatment, which could correspond to a more feasible intervention in practice.
\end{remark}

\begin{remark}
    Two additional complications addressed by these examples are: (1) the treatment is unordered, and (2) each parameter targets a specific treatment value. Typical dynamic stochastic interventions rely on the ordering of the treatment variable, considering only upward or downward shifts in treatment receipt, rather than targeting specific values. For instance, the definition of exponential tilts for continuous treatments implicitly relies on the ordering of the real line. Moreover, exponential tilts are not tailored to target a specific treatment value; rather, they  shift the treatment distribution either towards the maximum or minimum treatment. Since we focus on unordered categorical treatments and aim to target each treatment level separately, we must construct a suitable intervention. This is achieved by defining the intervention for all treatments $b \neq a$, while implicitly defining the intervention at treatment $a$.
\end{remark}

\section{General identification and estimation} \label{id-est:extra}

In this section, we develop general  methods for estimating sets of parameters $\left\{ \bbE \left( Y^{Q_a} \right) \right\}_{a=1}^{d}$ satisfying $V$-comparability and which are identifiable under \eqref{eq:positivity-condition} in the main paper. We assume general interventions as in example~\ref{example:smooth-general}. We begin with identification.
\begin{proposition} \label{prop:id-general} \textbf{(Identification)}
    Let $\big\{ \psi_a \big\}_{a=1}^{d} = \big\{ \bbE(Y^{Q_a}) \big\}_{a=1}^{d}$ denote a set of causal parameters defined by dynamic stochastic interventions as in example~\ref{example:smooth-general}, which vary with covariates $V \subseteq X$.  Suppose consistency and exchangeability hold, and the positivity assumption in \eqref{eq:positivity-condition} in the main paper holds for $V$. Then, for all $a \in \{1, \dots, d\}$,
    $$
    \psi_a = \bbE \left[ \sum_{b=1}^{d} \bbE \{ \mu_b(X) \mid V \} q_a(A =b \mid V) \right].
    $$
\end{proposition}

The next result provides the EIF for $\bbE(Y^{Q_a})$ with generic intervention function $f(\cdot)$ and smooth approximation function $s(\cdot)$.

\begin{theorem} \label{thm:eif-general} \textbf{(Efficient influence function)} 
    Let $\psi_a$ denote the identified parameter in Proposition~\ref{prop:id-general}, where $q_a$ is defined in example~\eqref{example:smooth-general} and $V = X$. Suppose Assumption~\ref{asmp:intermediate-positivity} holds and $f(\cdot)$ and $s(\cdot)$ are twice differentiable with non-zero bounded second derivatives. Then, the un-centered EIF of $\psi_a$ is
    \begin{align} 
        \varphi_a (Z) = \sum_{b=1}^{d} \Bigg( &\mu_b(Z) q_a(A=b \mid X) \nonumber \\
        &+ \left[ \frac{\one(A=b)}{\pi_b(X)}\left\{ Y - \mu_b(X) \right\} \right] q_a(A=b \mid X) + \mu_b(X) \varphi_{q_a}(Z; b) \Bigg) \label{eq:eif-general}
    \end{align}
    where 
    \begin{align}
        \varphi_{q_a}(Z; b) &= \varphi_S(Z) \{ \rho_a(A=b \mid X) - \pi_b(X) \} + S(X \in C_X) \varphi_{\rho_a}(Z; b) \nonumber  \\
        &\hspace{0.5in}+ \Big\{ 1 -S(X \in C_X) \Big\} \Big\{ \one(A=b) - \pi_b(X) \Big\}, \nonumber \\
        \varphi_{\rho_a}(Z; b) &=  \one(b \neq a) f^\prime \{ \pi_b(X) \}  \left\{ \one (A=b) - \pi_b(X) \right\} \nonumber \\
        &\hspace{1.2in} - \one(b=a)  \bigg[ \sum_{b \neq a} f^\prime \{ \pi_b(X) \}  \left\{ \one (A=b) - \pi_b(X) \right\} \bigg]  \text{, and } \label{eq:if-rho_a-general} \\
        \varphi_{S}(Z) &= \sum_{b=1}^{d} \bigg( s^\prime\{ \pi_b(X) \} \left\{ \one(A=b) - \pi_b(X) \right\} \bigg) \prod_{c \neq b}^{d} s\{ \pi_c(X) \}, \nonumber
    \end{align}
    and $\rho_a$ and $S(X \in C_X)$ are defined in \eqref{eq:trimmed-set}, in the main paper, and \eqref{eq:def-rho}, respectively.
\end{theorem}

The un-centered EIF in \eqref{eq:eif-general} is more complex than typical EIFs for standard estimands.  However, at a high level, it takes the usual form of a plug-in plus the sum of weighted residuals, as $\varphi_{q_a}(Z; b)$, $\varphi_{\rho_a}(Z; b)$, and $\varphi_S(Z)$ each consist of sums of weighted residuals.  It is possible to construct a one-step estimator for $\psi_a$ with the sample average of the un-centered EIF. Adding the weighted residuals to the plug-in debiases the plug-in estimate. The EIF could also be used to construct other efficient estimators, such as a targeted maximum likelihood estimator \citep{van2006targeted}. The one-step estimator we consider is also referred to as a double machine learning estimator \citep{chernozhukov2018double}. We focus on the one-step estimator because it has the same asymptotic guarantees as other estimators and is arguably simpler to construct.  The next result shows the one-step estimator's desirable asymptotic properties: its bias is doubly robust-style, in the sense that it is upper bounded by the sum of the products of errors from the nuisance function estimators.

\begin{theorem} \label{thm:dr-est-general} 
    \textbf{(Second order bias)} Let $\psi_a$ denote the identified parameter in Proposition~\ref{prop:id}, where $q_a$ is defined in example~\ref{example:smooth-general} and $V = X$. Suppose access to nuisance function estimates $\widehat \pi_b(X)$ and $\widehat \mu_b(X)$ for $b \in \{1,\dots, d\}$ independent from the observed sample, and construct a one-step estimator for $\psi_a$ as $\widehat \psi_a = \bbP_n \left\{ \widehat{\varphi}_{a}(Z) \right\}$, where  $\varphi_{a}(Z)$ is defined in \eqref{eq:eif-general} in thie section. Suppose further that
    \begin{enumerate}
        \item $f(\cdot)$ and $s(\cdot)$ are twice differentiable with non-zero bounded second derivatives and
        \item  $\left| \widehat \mu_b(X) \right|$ is uniformly bounded for all $b \in \{1, \dots, d \}$. 
    \end{enumerate}
    Then, 
    \begin{align}
        &\left| \bbE \left( \widehat \psi_a - \psi_a \right) \right| \lesssim \sum_{b=1}^{d} \left|  \bbE \left[ \left\{ \widehat \pi_b(X) - \pi_b(X) \right\} \left\{ \widehat \mu_b(X) - \mu_b(X) \right\}  \frac{\widehat q_a(A=b \mid X)}{\widehat \pi_b(X)} \right] \right| \nonumber \\\
        &+ \sum_{b=1}^{d} \left\lVert  \widehat \mu_b - \mu_b \right\rVert \Bigg[ \sum_{c=1}^d \left\lVert s^\prime(\pi_c) ( \widehat \pi_c - \pi_c ) \right\rVert  + \one(b \neq a) \lVert f^\prime(\pi_b) ( \widehat \pi_b - \pi_b ) \rVert + \one(b = a) \bigg\{ \sum_{b \neq a} \lVert f^\prime(\pi_b) ( \widehat \pi_b - \pi_b )  \rVert \bigg\} \Bigg] \nonumber \\
        &+ \Bigg\{ \sum_{b=1}^{d} \left\lVert \widehat \pi_b - \pi_b \right\rVert +  \sum_{b \neq a} \left\lVert f^\prime(\pi_b)  ( \widehat \pi_b - \pi_b ) \right\rVert \Bigg\} \left\{ \sum_{c=1}^{d}  \left\lVert s^\prime(\pi_c)  ( \widehat \pi_c - \pi_c ) \right\rVert \right\} \nonumber \\
        &+ \sum_{b \neq a}^{d} \left\lVert f^{\prime \prime} ( \pi_b )^{1/2} ( \widehat \pi_b - \pi_b) \right\rVert^2 \nonumber \\
        &+ d \Bigg\{ \sum_{b=1}^d \left\lVert s^{\prime \prime} ( \pi_b )^{1/2} ( \widehat \pi_b - \pi_b) \right\rVert^2 + \sum_{b=1}^{d} \sum_{c < b} \left\lVert s^\prime (\pi_b) (\widehat \pi_b - \pi_b ) \right\rVert \left\lVert s^\prime (\pi_c) (\widehat \pi_c - \pi_c ) \right\rVert  \Bigg\} \label{eq:dr-est-error-1} 
    \end{align}
\end{theorem}

This result shows that the one-step estimator has a second order bias consisting of sums of products of errors in estimating the propensity score and outcome regression. We first describe the conditions required for this result. The result assumes access to independent nuisance estimates. While this assumption is not strictly necessary for analyzing the bias of the estimator, it is essential for establishing its limiting distribution in Corollary~\ref{cor:normal}, as it allows us to avoid imposing Donsker or other complexity conditions \citep{van1996weak, chen2022debiased}. Moreover, this assumption is mild and can be ensured through cross-fitting, which typically involves randomly splitting the data into $k$ folds (five and ten are common), training the nuisance function estimators on $k-1$ folds, and evaluating the one-step estimator on the $k^{th}$ held-out fold \citep{robins2008higher, zheng2010asymptotic}. Full-sample efficiency can be retained by rotating the folds. The remaining two conditions are mild boundedness conditions; bounded non-zero derivatives for $f(\cdot)$ and $s(\cdot)$ can be enforced through the choice of intervention and smooth approximation function for the smoothed trimming indicator.
    
\smallskip

Next, we describe the bias term itself, in \eqref{eq:dr-est-error-1}. It consists of five pieces.  The first summand arises from estimating $\psi_a$ if the interventional propensity scores $q_a(A=b \mid X)$ were known, and is the canonical product of errors in estimating the nuisance functions. It could be upper bounded by a product of root-mean-squared-errors using H\"{o}lder's inequality and the Cauchy-Schwartz inequality, supposing $\left| \frac{\widehat q_a(A=b \mid X)}{\widehat \pi_b(X)} \right|$ were uniformly bounded for all $b \in \{1, \dots, d\}$.  We leave it as it is to highlight how different interventions could affect the error bound through $\widehat q_a(A=b \mid X)$. The second summand is the product of residuals in estimating $q_a(A=b \mid X)$ and $\mu_b(X)$, while the third summand is the product of residuals in estimating the two pieces of $q_a(A=b \mid X)$: $\rho_a(A=b \mid X)$ and $S(X \in C_X)$. The fourth summand arises from estimating $\rho_a(A=b \mid X)$ with a doubly robust-style estimator and the final summand arises from estimating $S(X \in C_X)$ with a doubly robust-style estimator. The next results examines how the bias can simplify under certain conditions.

\begin{corollary} \label{cor:dr-est-error-general}
    Under the conditions of Theorem~\ref{thm:dr-est}, suppose the trimmed set $C_X$ were known.  Then, $q_a(A=b \mid X) = \one(X \in C_X) \rho_a(A=b \mid X) + \one(X \notin C_X) \pi_b(X)$ and 
    \begin{align}
        &\left| \bbE \left( \widehat \psi_a - \psi_a \right) \right| \lesssim \sum_{b=1}^{d} \left|  \bbE \left[ \left\{ \widehat \pi_b(X) - \pi_b(X) \right\} \left\{ \widehat \mu_b(X) -  \mu_b(X) \right\}  \frac{\widehat q_a(A=b \mid X)}{\widehat \pi_b(X)} \right] \right| \nonumber \\
        &+ \sum_{b=1}^{d} \left\lVert  \widehat \mu_b - \mu_b \right\rVert \Bigg[ \one(b \neq a) \lVert f^\prime(\pi_b) ( \widehat \pi_b - \pi_b ) \rVert + \one(b = a) \bigg\{ \sum_{b \neq a} \lVert f^\prime(\pi_b) ( \widehat \pi_b - \pi_b )  \rVert \bigg\} \Bigg] \nonumber \\
        &+ \sum_{b \neq a}^{d} \left\lVert f^{\prime \prime} ( \pi_b )^{1/2} ( \widehat \pi_b - \pi_b) \right\rVert^2 \label{eq:dr-est-error-2} 
    \end{align}
    Meanwhile, if $f(\pi)$ were known without knowledge of the propensity scores (e.g., for trimmed treatment-specific means, $f(x) = 0$) but $S(X \in C_X)$ were unknown, then
    \begin{align}
        &\left| \bbE \left( \widehat \psi_a - \psi_a \right) \right| \lesssim \sum_{b=1}^{d} \left|  \bbE \left[ \left\{ \widehat \pi_b(X) - \pi_b(X) \right\} \left\{ \widehat \mu_b(X) - \mu_b(X) \right\}  \frac{\widehat q_a(A=b \mid X)}{\widehat \pi_b(X)} \right] \right| \nonumber \\
        &+ \sum_{b=1}^{d} \left( \left\lVert  \widehat \mu_b - \mu_b \right\rVert + \left\lVert \widehat \pi_b - \pi_b \right\rVert \right) \left\{ \sum_{c=1}^d \left\lVert s^\prime(\pi_c) ( \widehat \pi_c - \pi_c ) \right\rVert  \right\} \nonumber \\
        &+ d \Bigg\{ \sum_{b=1}^d \left\lVert s^{\prime \prime} ( \pi_b )^{1/2} ( \widehat \pi_b - \pi_b) \right\rVert^2 + \sum_{b=1}^{d} \sum_{c < b} \left\lVert s^\prime (\pi_b) (\widehat \pi_b - \pi_b ) \right\rVert \left\lVert s^\prime (\pi_c) (\widehat \pi_c - \pi_c ) \right\rVert  \Bigg\} \label{eq:dr-est-error-3} 
    \end{align}
    Finally, if both the trimmed set and $f(\pi)$ were known, then the estimand simplifies to \( \psi_a = \bbE \{ \mu_a(X) \one(X \in C_X) \} + \bbE\{ Y \one(X \notin C_X) \} \). Then, the doubly robust-style estimator $\widehat \psi_a := \bbP_n \left( \left[ \frac{\one(A=a)}{\widehat \pi_a(X)} \{ Y - \widehat \mu_a(X) \} + \widehat \mu_a(X) \right] \one(X \in C_X) + Y \one (X \notin C_X) \right)$ satisfies
    \begin{align}
        \left| \bbE \left( \widehat \psi_a - \psi_a \right) \right| &\lesssim \Big| \bbE \left[ \{ \widehat \pi_a(X) - \pi_a(X) \} \{ \widehat \mu_a(X) - \mu_a(X) \} \one \left( X \in C_X \right) \right] \Big|. \label{eq:dr-est-error-4}
    \end{align}
\end{corollary}

Corollary~\ref{cor:dr-est-error-general} demonstrates how the bias simplifies when the trimmed set $C_X$ is known and when the intervention function $f(\cdot)$ is known. The first result, in \eqref{eq:dr-est-error-2}, shows that the strong dependence on estimating the propensity scores disappears when the trimmed set is known. However, there is still some direct dependence on estimating the propensity scores, in order to construct the interventional propensity scores in $\rho_a(A=b \mid X)$. The result in \eqref{eq:dr-est-error-2} resembles typical results for dynamic stochastic interventions without trimming, such as for estimating IPSIs \citep{kennedy2019nonparametric}. Meanwhile, the second result in \eqref{eq:dr-est-error-3} establishes an upper bound on the bias when the trimmed set is unknown but the interventional propensity scores are known. This is the same result as Theorem~\ref{thm:dr-est}. Finally, \eqref{eq:dr-est-error-4} shows how everything simplifies if both the trimmed set and interventional propensity scores were known. Indeed, we return to the canonical doubly robust error (over the trimmed set) \citep{bang2005doubly}.

\color{black}
\section{Comparability and positivity proofs} \label{app:comparability-proofs}

\subsection*{Proof of Theorem~\ref{thm:positivity}}

\begin{proof}
    Lemma~\ref{lem:dynamic} establishes that criterion~\ref{crit:comp} is equivalent to properties~\ref{property:increase-a-conditional} and \ref{property:non-ab-equal}. Meanwhile, Lemma~\ref{lem:positivity} establishes the positivity assumption required to satisfy properties~\ref{property:increase-a-conditional}-\ref{property:weak-positivity}. Notice also that the positivity requirement in Lemma~\ref{lem:positivity} implies $\bbP \{ 0 < \pi_a(V) < 1 \} > 0$ for all $a \in \{1, \dots, d\}$; i.e., it implies Proposition~\ref{prop:positivity} (so there is no separate positivity requirement from Proposition~\ref{prop:positivity}).  
\end{proof}

\subsection*{Proof of Lemma~\ref{lem:dynamic}}

\begin{proof}
    $(\implies)$  Notice that 
    \begin{align}
        \bbE(Y^{Q_a}) - \bbE(Y^{Q_b}) &=  \bbE \left[ \sum_{c \in \mathcal{A}}  \bbE (Y^c \mid V) \{ q_a(c \mid V) - q_b(c \mid V) \} \right] \nonumber \\
        &\hspace{-0.5in}= \bbE \Big[ \bbE(Y^a \mid V) \{ q_a(a \mid V) - q_b(a \mid V) \} + \bbE(Y^b \mid V) \{ q_a(b \mid V) - q_b(b \mid V) \} \Big] \nonumber \\
        &\hspace{-0.5in}= \bbE \Big[ \big\{ \bbE(Y^a \mid V) - \bbE(Y^b \mid V) \big\} \big\{ q_a(a \mid V) - q_b(a \mid V) \big\} \Big] \label{eq:decomp}
    \end{align}   
    where the first line follows by iterated expectations on $V$ and the definition of the interventions and the second by property~\ref{property:non-ab-equal}. The third line, \eqref{eq:decomp}, follows again by property~\ref{property:non-ab-equal}. Because $q_a(c \mid V) = q_b(c \mid V)$ almost surely for all $c \notin \{ a, b\}$, this implies $q_a(a \mid V) + q_a(b \mid V) = 1 - \sum_{c \notin \{a, b\}} q_a(c \mid V) = 1 - \sum_{c \notin \{a, b\}} q_b(c \mid V) = q_b(a \mid V) + q_b(b \mid V)$ almost surely. Then, also notice that
    \begin{align*}
        &q_b(b \mid V) + q_b(a \mid V) = q_a(b \mid V) + q_a(a \mid V) \\
        \implies &q_b(b \mid V) - q_a(b \mid V) = - \{ q_b(a \mid V) - q_a(a \mid V) \}.
    \end{align*} 
    Next, recall that property~\ref{property:increase-a-conditional} guarantees $\bbP\{ q_a(A = a \mid V) > \bbP(A = a \mid V) \} > 0$ while $\bbP\{ q_b(A = a \mid V) \leq \bbP(A=a \mid V) \} = 1$. Therefore, $\bbP \{ q_a(A = a \mid V) > q_b(A = a \mid V) \} > 0$.  Moreover, note that also by property~\ref{property:increase-a-conditional} applied to $q_a$, $\bbP \{ q_a(A =b \mid V) \leq \bbP(A=b \mid V) \} = 1 \text{ for all } b \neq a$ which implies $\bbP\{ q_a(A=a \mid V) \geq \bbP(A=a \mid V) \} = 1$.  This, combined with $\bbP \{ q_b(A=a \mid V) \leq \bbP(A=a \mid V) \} = 1$, which follows by property~\ref{property:increase-a-conditional} applied to $q_b$, implies $\bbP \{ q_a(A = a \mid V) \geq q_b(A = a \mid V) \} = 1$. In summary,
    \begin{itemize}
        \item $\bbP \{ q_a(A = a \mid V) > q_b(A = a \mid V) \} > 0$ and
        \item $\bbP \{ q_a(A = a \mid V) \geq q_b(A = a \mid V) \} = 1$.
    \end{itemize}
    Hence, the sign of $\bbE(Y^{Q_a}) - \bbE(Y^{Q_b})$ or its equality to zero is inherited directly from \eqref{eq:cond_lower}-\eqref{eq:cond_upper} and criterion~\ref{crit:comp} holds.

    \bigskip

    \noindent $(\impliedby)$ The ``only if'' direction follows by contrapositive --- if $A$ and $B$ are events and $A^c$ and $B^c$ are their complements, then $(A^c \implies B^c) \implies (B \implies A)$.  If property~\ref{property:non-ab-equal} does not hold then there exists $c \notin \{a, b\}$ such that
    $$
    \bbP \{ q_a(A = c \mid V) = q_b(A = c \mid V) \} < 1.
    $$
    Hence, revisiting the line before \eqref{eq:decomp}, we have
    \begin{align*}
        \bbE(Y^{Q_a}) - \bbE(Y^{Q_b}) &= \bbE \Big[ \bbE(Y^a \mid V) \{ q_a(a \mid V) - q_b(a \mid V) \} + \bbE(Y^b \mid V) \{ q_a(b \mid V) - q_b(b \mid V) \} \Big] \\
        &+ \bbE \Big[ \bbE(Y^c \mid V) \big\{ q_a(c \mid V) - q_b(c \mid V) \Big].
    \end{align*}
    Then, criterion~\ref{crit:comp} no longer holds because it cannot rule out diabolical cases of distributions $\bbE(Y^c \mid V)$ that change the sign of $\bbE(Y^{Q_a}) - \bbE(Y^{Q_b})$. 
    
    \bigskip

    Meanwhile if property~\ref{property:non-ab-equal} holds but property~\ref{property:increase-a-conditional} does not, then \eqref{eq:decomp} still holds so that $\bbE(Y^{Q_a}) - \bbE(Y^{Q_b}) = \bbE \Big[ \big\{ \bbE(Y^a \mid V) - \bbE(Y^b \mid V) \big\} \big\{ q_a(b \mid V) - q_b(b \mid V) \big\} \Big]$, but one cannot rule out the case where $q_a(b \mid V) < q_b(b \mid V)$. Hence, criterion~\ref{crit:comp} does not hold.  
\end{proof}

\subsection*{Proof of Proposition~\ref{prop:positivity}}

\begin{proof}
    We prove the first statement by contrapositive. If $\bbP(A = a) = 1$ then $\bbP \{ \pi_a(V) = 1 \} = 1$. Therefore, $\bbP \{ q_a(A = a \mid V) > \pi_a(V) \} = 0$ and property~\ref{property:increase-a-conditional} cannot hold. Hence, if property~\ref{property:increase-a-conditional} holds, then it must be the case that $\bbP(A = 1) < 1$.
   
   \bigskip

    We prove the next statement using the following boolean logic: let $A$, $B$, and $C$ be three events and $A^c, B^c, C^c$ denote their complements. If $A \cap C^c \implies B^c$ and $B \cap C^c \implies A^c$, then $A \cap B \implies C$.

   \bigskip
   
   Suppose property~\ref{property:increase-a-conditional} holds and $\bbP \{ \bbP(A = a \mid X) = 0 \} = 1$.  By property~\ref{property:increase-a-conditional}, $\bbP \{ q_a(A = a \mid X) > 0 \} > 0$, which means property~\ref{property:weak-positivity} does not hold. Meanwhile, suppose property~\ref{property:weak-positivity} holds and $\bbP\{ \bbP(A = a \mid X) = 0 \} = 1$, then it must be the case that $\bbP \{ q_a(A = a \mid X) = 0 \} = 1$, so property~\ref{property:increase-a-conditional} does not hold.  Hence, if both properties hold, it must be the case that $\bbP \{ \bbP(A = a \mid X) > 0 \} > 0$.  Therefore, $\bbP(A=a) > 0$.
\end{proof}

\subsection*{Proof of Lemma~\ref{lem:positivity}}

\begin{proof}
    First, we establish some preliminary algebra. Let $Q_a, Q_b$ denote two arbitrary interventions and suppose properties~\ref{property:increase-a-conditional}-\ref{property:weak-positivity} hold. If $\bbP(A = a \mid X) = 0$, then
    $$
    q_a(A = a \mid V) = 0 = q_b(A = b \mid V)
    $$
    by property~\ref{property:weak-positivity}. Next, by property~\ref{property:non-ab-equal}, which guarantees $q_b(A = b \mid V) + q_b(A = a \mid V) = q_a(A = b \mid V) + q_a(A = a \mid V)$,
    $$
    q_b(A = b \mid V) = q_a(A = b \mid V).
    $$
    By property~\ref{property:increase-a-conditional}, which asserts $q_b(A = b \mid V) \geq \pi_b(V)$ and $q_a(A = b \mid V) \leq \pi_b(V)$,  
    $$
    q_b(A = b \mid V) = q_a(A = b \mid V) = \pi_b(V).
    $$
    Hence, when $\pi_a(V) = 0$, there is no intervention at $A =b$ or $A =a$. By properties~\ref{property:increase-a-conditional} and \ref{property:non-ab-equal}, there is also no intervention at $A = c$ for $c \notin \{a, b\}$.  In other words,
    $$
    \pi_a(V) = 0 \implies q_a(A = c \mid V) = q_b(A = c \mid V) = \bbP(A = c \mid V) \text{ for all } c \in \{1, \dots, d\}.
    $$
    By the same argument,
    $$
    \pi_b(V) = 0 \implies q_a(A = c \mid V) = q_b(A = c \mid V) = \bbP(A = c \mid V) \text{ for all } l \in \{1, \dots, d\}.
    $$

    \bigskip

    \noindent After this preliminary algebra, we prove the result via contradiction. Notice that 
    $$
    \bbE \left( \prod_{c \in \{a, b\}} \one \Big[ \bbP \left\{ \pi_c(X) > 0 \mid V \right\} = 1 \Big]\right) = 0
    $$
    is the complement of \eqref{eq:positivity-condition} for two interventions. If $\bbE \left( \prod_{c \in \{a, b\}} \one \Big[ \bbP \left\{ \pi_c(X) > 0 \mid V \right\} = 1 \Big]\right) = 0$, then there does not exist a set $C \in \mathcal{V}$ with $\bbP(V \in C) > 0$ such that $\bbP(A =a \mid v, w) > 0$ for all $v \in C$ and $w \in \mathcal{X} \setminus \mathcal{V}$.  Therefore, by the argument above, $\bbP \{ \bbP(A_c \mid V) = q_a(A = c \mid V) \} = 1$ for all $c \in \{1, \dots, d\}$. In other words, the interventions all leave the propensity scores unchanged, and therefore property~\ref{property:increase-a-conditional} is violated.

    \bigskip

    Hence, we have reached a contradiction --- if properties~\ref{property:increase-a-conditional}-\ref{property:weak-positivity} hold and 
    $$
    \bbE \left( \prod_{c \in \{a, b\}} \one \Big[ \bbP \left\{ \pi_c(X) > 0 \mid V \right\} = 1 \Big]\right) = 0,
    $$
    then property~\ref{property:increase-a-conditional} cannot hold. It follows that 
    $$
    \bbE \left( \prod_{c \in \{a, b\}} \one \Big[ \bbP \left\{ \pi_c(X) > 0 \mid V \right\} = 1 \Big]\right) > 0
    $$
    is necessary for properties~\ref{property:increase-a-conditional}-\ref{property:weak-positivity} to hold.

    \bigskip
        
    \noindent The result in Lemma~\ref{lem:positivity} follows by applying this proof to all pairs $(a,b) \in \{1, \dots, d\} \times \{ 1, \dots, d\}$.
\end{proof}

\section{Identification, efficient influence functions, and doubly robust-style estimators} \label{app:eif-estimator}

Proposition~\ref{prop:id-general} establishes the identification of the parameter. This follows immediately by consistency, exchangeability, $q$-weak positivity, and the standard g-formula.  Therefore, we derive the efficient influence function (EIF) for $\psi_a = \bbE\left(Y^{Q_a} \right) = \bbE \left[ \sum_{b \in \{1, \dots, d\}} \bbE \left\{ \mu_b(X) \mid V \right\} q_a(A =b \mid V) \right]$ where $q_a(A=b \mid V)$ is defined in \eqref{eq:q-def}. Initially, we demonstrate that the candidate EIFs for $\bbE \{ \mu(X) \mid V \}$ and $q_a(A=b \mid V)$ correspond to second-order remainder terms. Then, we combine these results to derive the EIF of $\psi_a$.   

\bigskip

Throughout what follows, we omit $b$ subscript notation unless it is necessary for clarification, so that $\mu_b \equiv \mu$ and $\pi_b \equiv \pi$. Moreover, we let $\overline \bbP$ denote another distribution in the space of distributions, $\mathcal{P}$. And, we denote nuisance functions from $\overline \bbP$ with ``overlines''; e.g., $\overline \mu_b(X)$ is the outcome regression function in the distribution $\overline \bbP$.

\begin{lemma} \label{lem:reg-eif}
    Let
    \begin{equation} \label{eq:if-mu}
        \varphi_{\mu}(Z; V, b) = \left[ \frac{\one(A=b)}{\pi_b(X)}\left\{ Y - \mu_b(X) \right\} + \mu_b(X) - \bbE \{ \mu_b(X) \mid V \} \right].
    \end{equation}
    Then,
    $$
    \bbE \left[ \overline \varphi_{\mu}(Z; V, b) + \overline \bbE \{ \overline \mu_b(X) \mid V\} - \bbE \{ \mu_b(X) \mid V \} \mid V \right] = \bbE \left[ \left\{ \frac{\pi_b(X) - \overline \pi_b(X)}{\overline \pi_b(X)} \right\} \left\{ \mu_b(X) - \overline \mu_b(X) \right\} \mid V  \right]
    $$
\end{lemma}

\begin{proof}
    We have
    \begin{align*}
        &\bbE \left[ \overline \varphi_{\mu}(Z; V, b)  + \overline \bbE \{ \overline \mu(X) \mid V \} - \bbE \{ \mu(X) \mid V \} \mid V \right] \\
        &= \bbE \left[ \frac{\one(A=b)}{\overline \pi(X)}\left\{ Y - \overline \mu(X) \right\} + \overline \mu(X) \mid V \right] - \overline \bbE \{ \overline \mu(X) \mid V \}  + \overline \bbE \{ \overline \mu(X) \mid V \} - \bbE \{ \mu(X) \mid V \}   \\
        &= \bbE \left[ \left\{ \frac{\pi(X) - \overline \pi(X)}{\overline \pi(X)} \right\} \left\{ \mu(X) - \overline \mu(X)  \right\} \mid V \right] \\
        &+ \bbE \left\{ \mu(X) - \overline \mu(X) + \overline \mu(X) \mid V \right\} - \overline \bbE \{ \overline \mu(X) \mid V) + \overline \bbE \{ \overline \mu(X) \mid V) - \bbE \{ \mu(X) \mid V \}.
    \end{align*}
    where the second line follows by iterated expectations on $X, \one(A=b)$ and iterated expectations on $X$, and adding and subtracting $\mu(X) - \overline \mu(X)$ inside the expectation. The result follows because the final line in the above display cancels out and equals zero.
\end{proof}

\noindent In Appendix~\ref{app:intermediate-algebra}, we address the intermediate algebra required to derive Lemma~\ref{lem:Q-eif}.

\begin{lemma} \label{lem:Q-eif}
    Let 
    \begin{align}
        \varphi_{q_a}(Z; V, b) &= \left\{ 1 -S(V \in C_V) \right\} \varphi_{\pi}(Z;V, b) - \varphi_S(Z; V) \pi_b(V) \nonumber \\
        &+ \varphi_S(Z; V) \rho_a(A=b \mid V) + S(V \in C_V) \varphi_{\rho_a}(Z; V, b) \label{eq:if-q} 
    \end{align}
    where $S(V \in C_V)$, $\varphi_{\pi}$, $\rho_a(A=b \mid V)$, $\varphi_{\rho_a}$, and $\varphi_S$ are defined in \eqref{eq:smooth-trim} (in the main paper), \eqref{eq:if-pi}, \eqref{eq:def-rho}, \eqref{eq:if-rho_a} (in Appendix~\ref{app:comp-extras}), and \eqref{eq:if-S}, respectively.  Suppose $s(\cdot)$ and $f(\cdot)$ are twice differentiable functions with bounded non-zero second derivatives. Then, omitting $X$ arguments,
    \begin{align*}
        &\bbE \left\{ \overline \varphi_{q_a}(Z; V, b) \mid V \right\} + \overline q_a(A=b \mid V) - q_a(A=b \mid V) \\
        &= \left\{ \pi_b(V) - \overline \pi_b(V) \right\} \left\{S(V \in C_V) - \overline S(V \in C_V) \right\} \\
        &+ \left\{ \overline \rho_a(b \mid V) - \overline \pi_b(V) \right\} \Bigg\{ \sum_{a=1}^{d} \prod_{b \neq a}^{d} \overline \bbE \{ s (\overline \pi_b) \mid  V \} \Big( \bbE \Big[ \{ s^\prime(\overline \pi_a) - s^\prime (\pi_a) \} (\pi_a - \overline \pi_a) \\
        &\hspace{3in}+ \left\{ s^{\prime \prime}(\pi_a) + o(1)\right\} (\overline \pi_a - \pi_a)^2  \mid V \Big] \Big)  \\
        &+  \Bigg[ \bbE \{ s(\pi_a) \mid V \} - \overline \bbE \big\{ s( \overline \pi_a ) \mid V \big\} \Bigg] \Bigg[ \prod_{b < a} \overline \bbE \{ s( \overline \pi_b) \mid V \} - \prod_{b < a} \bbE \{ s(\pi_b) \mid V \} \Bigg] \Bigg[ \prod_{b > a} \overline \bbE \{ s (\overline \pi_b) \mid V \} \Bigg] \Bigg\}  \\
        &+S(V \in C_V) \Bigg\{ \one(b \neq a) \Big[ \big\{ f^{\prime \prime}\{ \pi_b(V)\} + o(1) \big\} \big\{ \overline \pi_b(V) - \pi_b(V) \big\}^2 \Big] \\
        &\hspace{0.85in}- \one(b = a) \bigg( \sum_{b \neq a} \big[ f^{\prime \prime}\{ \pi_b(V) \} + o(1) \big] \big\{ \overline \pi_b(V) - \pi_b(V) \big\}^2 \bigg)   \Bigg\}
    \end{align*}
    where $q_a(A=b \mid V)$ is defined in \eqref{eq:q-def} (in Appendix~\ref{app:comp-extras}).
\end{lemma}

\begin{proof}
    Omitting $V$ and $Z$ arguments from the second equality onwards below, 
    \begin{align*}
        &\bbE \{ \overline \varphi_{q_a}(Z; V, b) \mid V  \}  + \overline q_a(A=b \mid V) - q_a(A=b \mid V) \\
        &= \bbE \bigg[  \left\{ 1 - \overline S(V \in C_V) \right\} \overline \varphi_{\pi}(Z; V, b) - \overline \varphi_S(Z; V)  \overline \pi_b(V) \\
        &\hspace{0.25in}+ \overline \varphi_S(Z; V) \overline \rho_a(A=b \mid V) + \overline S(V \in C_V) \overline \varphi_{\rho_a}(Z; V, b) \mid V \bigg] \\
        &\hspace{0.25in}+ \left\{ 1 - \overline S(V \in C_V) \right\} \overline \pi_b(V) + \overline S(V \in C_V) \overline \rho_a(A=b \mid V) \\
        &\hspace{0.25in}- \left\{ 1 - S(V \in C_V) \right\} \pi_b(V) - S(V \in C_V) \rho_a(A=b \mid V) \\
        &= \bbE \Big\{  (1 - \overline S) \big( \overline \varphi_{\pi}(b) + \overline \pi_b - \pi_b \big) + \pi_b (S - \overline S) - \overline \varphi_S \overline \pi_b 
        + \overline \varphi_S \overline \rho_a(b) + S \overline \varphi_{\rho_a}(b) + \overline S \overline \rho_a (b) - S \rho_a (b) \mid V \Big\} \\
        &= \bbE \Big\{  0 + (\pi_b - \overline \pi_b) (S - \overline S) + \overline \pi_b \Big( S - \overline S - \overline \varphi_S \Big) + \overline \varphi_S \overline \rho_a(b) + S \overline \varphi_{\rho_a}(b) + \overline S \overline \rho_a (b) - S \rho_a (b) \mid V \Big\} \\
        &=\bbE \Big\{ (\pi_b - \overline \pi_b) (S - \overline S) + \overline \pi_b \Big( S - \overline S - \overline \varphi_S \Big) + \overline \rho_a(b) (\overline \varphi_S + \overline S - S) + S \Big( \overline \rho_a(b) + \overline \varphi_{\rho_a}(b) - \rho_a(b) \Big) \mid V  \Big\} \\
        &=(\pi_b - \overline \pi_b) (S - \overline S) + \Big\{ \overline \pi_b - \overline \rho_a(b) \Big\} \Big\{ S - \overline S - \bbE \Big( \overline \varphi_S \mid V \Big) \Big\} + S \Big\{ \overline \rho_a(b) + \bbE \big( \overline \varphi_{\rho_a}(b) \mid V \big) - \rho_a(b) \Big\}
    \end{align*}
    where each line follows by adding zero, the third equality follows because
    $$
    \bbE \big(  \overline \varphi_{\pi}(b) + \overline \pi_b - \pi_b \mid V \big) =  \bbE \Big\{ (1 - \overline S) \big( \overline \varphi_{\pi}(b) + \overline \pi_b - \pi_b \big) \mid V \Big\} = 0,
    $$
    and the final line follows because $\pi_b, S, \rho_a(b)$ are all constant conditional on $V$. 

    \bigskip 
    
    By the same argument as in the proof of Proposition~\ref{prop:S-eif},
    \begin{align*}
        &\big\{ \overline \rho_a(b) - \overline \pi_b \big\}  \Big\{ \bbE  \big( \overline \varphi_S \mid V \big)+ \overline S - S \big\}  \\
        &= \big\{ \overline \rho_a(b) - \overline \pi_b \big\} \Bigg\{ \sum_{a=1}^{d} \prod_{b \neq a}^{d} \overline \bbE \{ s (\overline \pi_b) \mid  V \} \left( \bbE \Big[ \{ s^\prime(\overline \pi_a) - s^\prime (\pi_a) \} (\pi_a - \overline \pi_a) + \left\{ s^{\prime \prime}(\pi_a) + o(1)\right\} (\overline \pi_a - \pi_a)^2  \mid V \Big] \right) \\
        &+ \Bigg[ \bbE \{ s(\pi_a) \mid v\} - \overline \bbE \big\{ s( \overline \pi_a ) \mid V \big\} \Bigg] \Bigg[ \prod_{b < a} \overline \bbE \{ s( \overline \pi_b) \mid v \} - \prod_{b < a} \bbE \{ s(\pi_b) \mid v\} \Bigg] \Bigg[ \prod_{b > a} \overline \bbE \{ s (\overline \pi_b) \mid V \} \Bigg] \Bigg\}.
    \end{align*}
    Finally, by the same argument as in the proof of Proposition~\ref{prop:f-eif}. 
    \begin{align*}
        \bbE &\Big[ S \Big\{ \overline \rho_a(b) + \overline \varphi_{\rho_a}(b) - \rho_a(b) \Big\} \mid V  \Big] \\
        &=  S \bbE \Bigg[ \one(b \neq a) \big\{ \overline \varphi_{f}(b) + f(\overline \pi_b) - f(\pi_b) \big\} - \one(b \neq a) \bigg\{ \sum_{b = a} \overline \varphi_{f}(b) + f (\overline \pi_b) - f(\pi_b) \bigg\} \mid V \Bigg]  \\
        &=  S \Bigg( \one(b \neq a) \Big[ \big\{ f^{\prime \prime}(\pi_b) + o(1) \big\} (\overline \pi_b - \pi_b)^2 \Big] - \one(b = a) \bigg\{ \sum_{b \neq a} \big\{ f^{\prime \prime}(\pi_b) + o(1) \big\} (\overline \pi_b - \pi_b)^2 \bigg\}   \Bigg) 
    \end{align*}
\end{proof}

Next, we address the EIF of $\psi_a$. The proof below shows that a relevant error term is second order. The fact that this implies the relevant function is the EIF follows by \citet[Lemma 2]{kennedy2023semiparametric}.  Generally, the proofs take the following form. Let $\varphi_\theta(Z)$ denote the centered EIF of a generic parameter $\theta: \bbP \to \bbR$. Then: 
\begin{enumerate}
    \item Show that $\bbE \{ \varphi_\theta(Z) \} = 0$.
    \item Show that $\bbE \big\{ \overline \varphi_\theta(Z) \big\} + \overline \theta - \theta$ is second order in terms of the relevant nuisance functions, where $\overline \varphi_\theta, \overline \theta$ are formed from nuisance estimates from the data generating process $\overline{\bbP}$.
\end{enumerate}

\begin{theorem} \label{thm:eif-overall} \textbf{(Efficient influence function)} Under the setup of Theorem~\ref{thm:eif}, the centered efficient influence function of $\psi_a$ is
    \begin{equation} \label{eq:if-psi-a}
        \phi_a (Z) = \sum_{b=1}^{d} \varphi_{\mu}(Z; V, b)  q_a(A=b \mid V) + \bbE \{ \mu_b(X) \mid V \} \varphi_{q_a}(Z; V, b) 
    \end{equation}
    where $q_a$, $\varphi_{\mu}$, and $\varphi_{q_a}$ are defined in \eqref{eq:q-def} (in Appendix~\ref{app:comp-extras}), \eqref{eq:if-mu}, and \eqref{eq:if-q}, respectively.
\end{theorem}

\begin{proof}
    First, notice that $\bbE \{ \phi_a(Z) \} = 0$ by iterated expectations.  Next,
    \begin{align*}
        \bbE \Bigg[ \overline \phi_a(Z) + &\sum_{b=1}^{d} \overline \bbE \{ \overline \mu_b(X) \mid V\} \overline q_a(A=b \mid V) - \bbE \{ \mu_b(X) \mid V \} q_a(A=b \mid V) \Bigg] \\
        = \bbE \Bigg[ &\sum_{b=1}^{d} \overline \varphi_{\mu}(Z; V, b) \overline q_a(A=b \mid V) +  \overline \bbE \{ \overline\mu_b(X) \mid V \} \overline \varphi_{q_a} (Z; V, b)  \\
        &+ \overline \bbE \{ \overline \mu_b(X) \mid V\} \overline q_a(A=b \mid V) - \bbE \{ \mu_b(X) \mid V \} q_a(A=b \mid V) \Bigg] \\
        = \sum_{b=1}^{d} &\bbE \bigg( \Big[ \overline \varphi_{\mu}(Z; V, b) + \overline \bbE \{ \overline \mu_b(X) \mid V \} - \bbE \{ \mu_b(X) \mid V \} \Big] \overline q_a(A=b \mid V) \bigg) \\
        + \bbE \Bigg( &\Big[ \overline \varphi_{q_a}(Z; V, b) + \overline q_a(A=b \mid V) - q_a(A=b \mid V) \Big] \overline \bbE \{ \overline \mu_b(X) \mid V \} \Bigg) \\
        + \bbE \bigg( &\Big\{ \overline q_a(A=b \mid V) - q_a(A=b \mid V) \Big\} \Big[ \bbE \{ \mu_b(X) \mid V \} - \overline \bbE \{ \overline \mu_b(X) \mid V \} \Big] \bigg) 
    \end{align*}
    where the first line follows by definition of $\phi_a$ and the second follows by adding and subtracting $\bbE \Big[ \bbE \{ \mu_b(X) \mid V \} \overline q_a(A=b \mid V) \Big]$ and $\bbE \Big[ \overline \bbE \{ \overline \mu_b(X) \mid V \} \Big\{ \overline q_a(A=b \mid V) - q_a(A=b \mid V) \Big\} \Big]$. The final term in the above display is second order by the same analysis as in the proofs of Lemmas~\ref{lem:reg-eif} and \ref{lem:Q-eif}.  We analyze it in further detail in the proof of Theorem~\ref{thm:dr-est-V}, which also demonstrates that it is second order.  The intermediate positivity assumption in Assumption~\ref{asmp:intermediate-positivity} is necessary so that $\phi_a(Z)$ has finite variance and therefore Lemma 2 from \citet{kennedy2023semiparametric} can be applied.
\end{proof}

\begin{theorem} \label{thm:dr-est-V} 
    \textbf{(Doubly robust estimator)} For $V \subseteq X$, suppose access to nuisance function estimates $\widehat \pi_a(V), \widehat \pi_a(X), f \{ \widehat \pi_a(V) \}$,
    $f^\prime\{ \widehat \pi_a(V)\}, s\{ \widehat \pi_a(X)\}, s^\prime \{ \widehat \pi_a(X) \}, \widehat \bbE\{ s\{ \widehat \pi_a(X)\} \mid V \}, \widehat \mu_a(X)$, and $\widehat \bbE \{ \widehat \mu_a(X) \mid V \}$ for $a=1$ to $d$ which are independent from the observed sample.  Then, construct an estimator for $\psi_a$ as $\widehat \psi_a = \bbP_n \{ \widehat \varphi_a(Z) \}$ where $\phi_a(Z)$ is the un-centered efficient influence function of $\psi_a$ given in Theorem~\ref{thm:eif}. Suppose further that 
    \begin{enumerate}
        \item $\frac{\widehat q_a(A=b \mid X)}{\widehat \pi_b(X)}$ and $\left| \widehat \mu_b(X) \right|$ are uniformly bounded for all $b \in \{1, \dots, d \}$ and 
        \item $s(\cdot)$ and $f(\cdot)$ are twice differentiable with non-zero bounded second derivatives.
    \end{enumerate}
    Then, the estimator is doubly robust, in the sense that is has a bias that is bounded by a product of nuisance estimator errors:
    \begin{align*}
        &\left| \bbE \left( \widehat \psi_a - \psi_a \right) \right| \lesssim \sum_{b=1}^{d} \left| \bbE \Bigg(  \bbE \left[ \left\{ \frac{\pi_b(X) - \widehat \pi_b(X)}{\widehat \pi_b(X)} \right\} \left\{ \mu_b(X) - \widehat \mu_b(X) \right\} \mid V \right] \widehat q_a(A=b \mid V) \Bigg)\right|  \\
        &+ d \Bigg[ \sum_{b=1}^d \left\lVert s^{\prime \prime}\{ \pi_b(X) \}^{1/2} \{ \widehat \pi_b(X) - \pi_b(X) \} \right\rVert^2 + \left\lVert \bbE\{ \widehat \pi_b(X) - \pi_b(X) \mid V \} \right\rVert \left\lVert \bbE \{ s^\prime\{ \widehat \pi_b(X) \} - s^\prime \{ \pi_b(X) \} \mid V \} \right\rVert \Bigg]  \\
        &+ d \Bigg[ \sum_{b=1}^{d} \sum_{c < b} \left\lVert \widehat \bbE \big[ s\{ \widehat \pi_b(X) \} \mid V \big] - \bbE \big[ s\{ \pi_b(X) \} \mid V \big] \right\rVert \left\lVert \widehat \bbE \big[ s\{ \widehat \pi_c(X) \} \mid V \big] - \bbE \big[ s\{ \pi_c(X) \} \mid V \big] \right\rVert \Bigg] \\
        &+ \Bigg[ \sum_{b=1}^{d} \left\lVert \widehat \pi_b(X) - \pi_b(X) \right\rVert +  \sum_{b \neq a} \left\lVert f \{ \widehat \pi_b(X) \} - f \{ \pi_b(X) \} \right\rVert \Bigg] \left( \sum_{c=1}^{d}  \left\lVert \widehat \bbE \big[ s\{ \widehat \pi_c(X) \} \mid V \big] - \bbE \big[ s\{ \pi_c(X) \} \mid V \big] \right\rVert \right)  \\
        &+ \sum_{b \neq a}^{d} \left \lVert f^\prime \{ \widehat \pi_b(V) \} - f^\prime \{ \pi_b(V) \} \right\rVert \left\lVert \widehat \pi_b(V) - \pi_b(V) \right\rVert + \left\lVert f^{\prime \prime} \{ \pi_b(V) \}^{1/2} \{ \widehat \pi_b(V) - \pi_b(V) \} \right\rVert^2 \\
        &+ \sum_{b=1}^{d} \left\lVert \widehat \bbE \{ \widehat \mu_b(X) \mid V \} - \bbE \{ \mu_b(X) \mid V \} \right\rVert \left\lVert \widehat \pi_b(V) - \pi_b(V) \right\rVert \\
        &+ \sum_{b=1}^{d} \left\lVert \widehat \bbE \{ \widehat \mu_b(X) \mid V \} - \bbE \{ \mu_b(X) \mid V \} \right\rVert \Bigg( \sum_{c=1}^d \left\lVert \widehat \bbE \big[ s\{ \widehat \pi_c(X) \} \mid V \big] - \bbE \big[ s\{ \pi_c(X) \} \mid V \big] \right\rVert \Bigg) \\
        &+ \sum_{b=1}^{d} \left\lVert \widehat \bbE \{ \widehat \mu_b(X) \mid V \} - \bbE \{ \mu_b(X) \mid V \} \right\rVert \Bigg( \one(b \neq a) \lVert f\{ \widehat \pi_b(V) \} - f\{ \pi_b(V)\} \rVert \\
        &\hspace{3in} + \one(b = a) \Big[ \sum_{b \neq a} \lVert f \{ \widehat \pi_b(V) \} - f\{ \pi_b(V) \} \rVert \Big]  \Bigg).
    \end{align*}
\end{theorem}

\noindent Before proving Theorem~\ref{thm:dr-est-V}, we state several corollaries, which imply Theorem~\ref{thm:dr-est} in the main text.

\begin{corollary} \textbf{(Simplifying with $V = X$)}
    Suppose the conditions of Theorem~\ref{thm:dr-est} are satisfied for $V = X$. Then, suppressing ubiquitous dependence on $X$, 
    \begin{align*}
        &\left| \bbE \left( \widehat \psi_a - \psi_a \right) \right| \lesssim \sum_{b=1}^{d} \left|  \bbE \left[ \left\{ \pi_b(X) - \widehat \pi_b(X) \right\} \left\{ \mu_b(X) - \widehat \mu_b(X) \right\}  \frac{\widehat q_a(A=b \mid X)}{\widehat \pi_b(X)} \right] \right| \\
        &+ d \Bigg\{ \sum_{b=1}^d \left\lVert s^{\prime \prime} ( \pi_b )^{1/2} ( \widehat \pi_b - \pi_b) \right\rVert^2 + \left\lVert \widehat \pi_b - \pi_b \right\rVert \left\lVert s^\prime ( \widehat \pi_b ) - s^\prime ( \pi_b ) \right\rVert + \sum_{b=1}^{d} \sum_{c < b} \left\lVert s ( \widehat \pi_b ) -  s ( \pi_b ) \right\rVert \left\lVert s ( \widehat \pi_c ) - s ( \pi_c ) \right\rVert \Bigg\} \\
        &+ \Bigg\{ \sum_{b=1}^{d} \left\lVert \widehat \pi_b - \pi_b \right\rVert +  \sum_{b \neq a} \left\lVert f ( \widehat \pi_b ) - f ( \pi_b ) \right\rVert \Bigg\} \left\{ \sum_{c=1}^{d}  \left\lVert s ( \widehat \pi_c ) - s ( \pi_c ) \right\rVert \right\}  \\
        &+ \sum_{b \neq a}^{d} \left \lVert f^\prime ( \widehat \pi_b ) - f^\prime ( \pi_b ) \right\rVert \left\lVert \widehat \pi_b - \pi_b \right\rVert + \left\lVert f^{\prime \prime} ( \pi_b )^{1/2} ( \widehat \pi_b - \pi_b) \right\rVert^2 \\
        &+ \sum_{b=1}^{d} \left\lVert  \widehat \mu_b - \mu_b \right\rVert \Bigg\{ \sum_{b=1}^d \left\lVert  s ( \widehat \pi_b ) - s ( \pi_b ) \right\rVert  + \one(b \neq a) \lVert f ( \widehat \pi_b ) - f( \pi_b ) \rVert + \one(b = a) \Big\{ \sum_{b \neq a} \lVert f ( \widehat \pi_b ) - f ( \pi_b )  \rVert \Big\} \Bigg\}.
    \end{align*}
    Moreover, simplifying terms via Taylor expansion (e.g, for $s^{\prime}(\widehat \pi_b) - s^\prime (\pi_b)$) yields
    \begin{align*}
        &\left| \bbE \left( \widehat \psi_a - \psi_a \right) \right| \lesssim \sum_{b=1}^{d} \left|  \bbE \left[ \left\{ \pi_b(X) - \widehat \pi_b(X) \right\} \left\{ \mu_b(X) - \widehat \mu_b(X) \right\}  \frac{\widehat q_a(A=b \mid X)}{\widehat \pi_b(X)} \right] \right| \\
        &+ d \Bigg\{ \sum_{b=1}^d \left\lVert s^{\prime \prime} ( \pi_b )^{1/2} ( \widehat \pi_b - \pi_b) \right\rVert^2 + \sum_{b=1}^{d} \sum_{c < b} \left\lVert s^\prime (\pi_b) (\widehat \pi_b - \pi_b ) \right\rVert \left\lVert s^\prime (\pi_c) (\widehat \pi_c - \pi_c ) \right\rVert  \Bigg\} \\
        &+ \Bigg\{ \sum_{b=1}^{d} \left\lVert \widehat \pi_b - \pi_b \right\rVert +  \sum_{b \neq a} \left\lVert f^\prime(\pi_b)  ( \widehat \pi_b - \pi_b ) \right\rVert \Bigg\} \left\{ \sum_{c=1}^{d}  \left\lVert s^\prime(\pi_c)  ( \widehat \pi_c - \pi_c ) \right\rVert \right\}  \\
        &+ \sum_{b \neq a}^{d} \left\lVert f^{\prime \prime} ( \pi_b )^{1/2} ( \widehat \pi_b - \pi_b) \right\rVert^2 \\
        &+ \sum_{b=1}^{d} \left\lVert  \widehat \mu_b - \mu_b \right\rVert \Bigg[ \sum_{c=1}^d \left\lVert s^\prime(\pi_c) ( \widehat \pi_c - \pi_c ) \right\rVert  + \one(b \neq a) \lVert f^\prime(\pi_b) ( \widehat \pi_b - \pi_b ) \rVert + \one(b = a) \bigg\{ \sum_{b \neq a} \lVert f^\prime(\pi_b) ( \widehat \pi_b - \pi_b )  \rVert \bigg\} \Bigg].
    \end{align*}
\end{corollary}

\begin{corollary} \textbf{(Upper bounds on derivatives of $s$ and $f$)}
    Under the conditions of Theorem~\ref{thm:dr-est} and with $V = X$, suppose the first and second derivatives of $s(\cdot)$ and $f(\cdot)$ are bounded with $s^\prime \leq M_{s^\prime}, s^{\prime \prime} \leq M_{s^\prime \prime}, f^\prime \leq M_{f^\prime}$, and $f^{\prime \prime} \leq M_{f^{\prime \prime}}$.  Then,
    \begin{align*}
        &\left| \bbE \left( \widehat \psi_a - \psi_a \right) \right| \lesssim \sum_{b=1}^{d} \left|  \bbE \left[ \left\{ \pi_b(X) - \widehat \pi_b(X) \right\} \left\{ \mu_b(X) - \widehat \mu_b(X) \right\}  \frac{\widehat q_a(A=b \mid X)}{\widehat \pi_b(X)} \right] \right| \\
        &+ d \Bigg\{ \sum_{b=1}^d M_{s^{\prime \prime}} \left\lVert   \widehat \pi_b - \pi_b \right\rVert^2 + \sum_{b=1}^{d} \sum_{c < b} M_{s^\prime}^2 \left\lVert \widehat \pi_b - \pi_b \right\rVert \left\lVert \widehat \pi_c - \pi_c  \right\rVert  \Bigg\} \\
        &+ \Bigg\{ \sum_{b=1}^{d} \sum_{c=1}^{d} M_{s^\prime} \left\lVert \widehat \pi_b - \pi_b \right\rVert \left\lVert \widehat \pi_c - \pi_c \right\rVert  +  \sum_{b \neq a} \sum_{c=1}^{d} M_{f^\prime} M_{s^\prime} \left\lVert \widehat \pi_b - \pi_b \right\rVert \left\lVert \widehat \pi_c - \pi_c \right\rVert  \\
        &+ \sum_{b \neq a}^{d} M_{f^{\prime \prime}} \left\lVert  \widehat \pi_b - \pi_b \right\rVert^2 \\
        &+ \sum_{b=1}^{d} \left\lVert  \widehat \mu_b - \mu_b \right\rVert \Bigg\{ \sum_{b=1}^d M_{s^\prime} \left\lVert  \widehat \pi_b - \pi_b \right\rVert  + \one(b = a) M_{f^\prime} \lVert \widehat \pi_b - \pi_b \rVert + \one(b \neq a) \sum_{b \neq a} M_{f^\prime} \lVert \widehat \pi_b - \pi_b  \rVert \Bigg\}.
    \end{align*}
    Ignoring these constants, the upper bound on the bias simplifies to
    \begin{align*}
        \left| \bbE \left( \widehat \psi_a - \psi_a \right) \right| &\lesssim \sum_{b=1}^{d} \left|  \bbE \left[ \left\{ \pi_b(X) - \widehat \pi_b(X) \right\} \left\{ \mu_b(X) - \widehat \mu_b(X) \right\}  \frac{\widehat q_a(A=b \mid X)}{\widehat \pi_b(X)} \right] \right| \\
        &+ \sum_{b=1}^{d} \sum_{c \neq b} \lVert \widehat \mu_b - \mu_b \rVert  \lVert \widehat \pi_c - \pi_c \rVert + d \sum_{b=1}^{d} \sum_{c \leq b}  \lVert \widehat \pi_b - \pi_b \rVert \lVert \widehat \pi_c - \pi_c \rVert.  
    \end{align*}
\end{corollary}

% \begin{corollary} \textbf{(Incorporating bounds on $s^\prime, s^{\prime \prime}$)}
%     Under the conditions of the previous result, suppose $s(x) = 1 - \exp(-k_nx)$ so that $|s^{\prime}(x)| \leq k_n, |s^{\prime \prime}(x)| \leq k_n^2$ for $x \in [0,1]$, and $|f^\prime|, |f^{\prime \prime}|$ are bounded by some constant that does not depend on sample size. Then, 
%     \begin{equation}
%         \left| \bbE \left( \widehat \psi_a - \psi_a \right) \right| \lesssim \sum_{b=1}^{d} \sum_{c \neq b} k_n \lVert \widehat \mu_b - \mu_b \rVert  \lVert \widehat \pi_c - \pi_c \rVert + d k_n^2 \Big( \lVert \widehat \pi_b - \pi_b \rVert \lVert \widehat \pi_c - \pi_c \rVert \Big). 
%     \end{equation}
% \end{corollary}

\bigskip
\noindent \textbf{Proof of Theorem~\ref{thm:dr-est-V}}

\begin{proof}
    That this error is second order follows directly from the proofs of Lemmas~\ref{lem:reg-eif} and \ref{lem:Q-eif}, replacing the overline's by hat's.  However, here we will provide the full expression of the bias in detail.  First, by the argument in the proof of Theorem~\ref{thm:eif-overall},
    \begin{align}
        \bbE \left( \widehat \psi_a - \psi_a \right) = \sum_{b=1}^{d} &\bbE \bigg( \Big[ \widehat \varphi_{\mu}(Z; V, b) + \widehat \bbE \{ \widehat \mu_b(X) \mid V \} - \bbE \{ \mu_b(X) \mid V \} \Big] \widehat q_a(A=b \mid V) \bigg) \nonumber \\
        + \bbE \Bigg( &\Big[ \widehat \varphi_{q_a}(Z; V, b) + \widehat q_a(A=b \mid V) - q_a(A=b \mid V) \Big] \widehat \bbE \{ \widehat \mu_b(X) \mid V \} \Bigg) \nonumber \\
         &\hspace{-0.75in}+ \bbE \bigg( \Big\{ \widehat q_a(A=b \mid V) - q_a(A=b \mid V) \Big\} \Big[ \bbE \{ \mu_b(X) \mid V \} - \widehat \bbE \{ \widehat \mu_b(X) \mid V \} \Big] \bigg) \label{eq:original-display}
    \end{align}
    where $\varphi_\mu$ and $\varphi_{q_a}$ are defined in \eqref{eq:if-mu} and \eqref{eq:if-q}. We'll consider each line of~\eqref{eq:original-display} in order. First, 
    \begin{align*}
        &\bbE \bigg( \Big[ \widehat \varphi_{\mu}(Z; V, b) + \widehat \bbE \{ \widehat \mu_b(X) \mid V \} - \bbE \{ \mu_b(X) \mid V \} \Big] \widehat q_a(A=b \mid V) \bigg) \\
        &= \bbE \Bigg(  \left[ \frac{\one(A=b)}{\widehat \pi_b(X)}\left\{ Y - \widehat \mu_b(X) \right\} + \widehat  \mu_b(X) - \bbE \{ \mu_b(X) \mid V \} \right] \widehat q_a(A=b \mid V) \Bigg) \\
        &= \bbE \Bigg(  \bbE \left[ \left\{ \frac{\pi_b(X) - \widehat \pi_b(X)}{\widehat \pi_b(X)} \right\} \left\{ \mu_b(X) - \widehat \mu_b(X) \right\} \mid V \right] \widehat q_a(A=b \mid V) \Bigg) 
    \end{align*}
    where the first line follows by definition, and the second by iterated expectations on $\one(A=b), X$, iterated expectations on $X$, adding and subtracting $\left\{ \mu_b(X) - \widehat \mu_b(X) \right\} \widehat q_a(A=b \mid V)$, iterated expectations on $V$, and cancelling terms. 

    \bigskip

    Next, we consider the second term in \eqref{eq:original-display},
    $$
    \bbE \Bigg( \Big[ \widehat \varphi_{q_a}(Z; V, b) + \widehat q_a(A=b \mid V) - q_a(A=b \mid V) \Big] \widehat \bbE \{ \widehat \mu_b(X) \mid V \} \Bigg)
    $$
    where $\varphi_{q_a}(Z; V, b)$ is defined in \eqref{eq:if-q}. Continuing from the proof of Lemma~\ref{lem:Q-eif}, 
    \small
    \begin{align}
        &\bbE \Bigg( \Big[ \widehat \varphi_{q_a}(Z; V, b) + \widehat q_a(A=b \mid V) - q_a(A=b \mid V) \Big] \widehat \bbE \{ \widehat \mu_b(X) \mid V \} \Bigg) \nonumber \\    
        &= \bbE \Bigg( \widehat \bbE \{ \widehat \mu_b(X) \mid V\} \big\{ \widehat \rho_a(A=b \mid V) - \widehat \pi_b(V) \big\} \bbE \Big\{  \widehat \varphi_S(Z) + \widehat S(V \in C_V) - S(V \in C_V)  \mid V \Big\} \Bigg) \label{eq:decomp-1} \\
        &+ \bbE \Bigg( \widehat \bbE \{ \widehat \mu_b(X) \mid V\}  \big\{ S(V \in C_V) - \widehat S(V \in C_V) \big\} \bbE \Big[ \big\{ \widehat \rho_a(A=b \mid V) - \widehat \pi_b(V) \big\} \nonumber \\
        &\hspace{3in}- \big\{ \rho_a(A=b \mid V) - \pi_b(V) \big\} \mid V \Big] \Bigg) \label{eq:decomp-2} \\
        &+ \bbE \Bigg( \widehat \bbE \{ \widehat \mu_b(X) \mid V) \} \widehat S(V \in C_V) \bbE \Big\{ \widehat \varphi_{\rho_a}(Z; V, b) + \widehat \rho_a(A=b \mid V) - \rho_a(A=b \mid V) \mid V \Big\} \Bigg), \label{eq:decomp-3} 
    \end{align}
    \normalsize
    where $\rho_a$ and $\varphi_S$ are defined in \eqref{eq:def-rho} (in Appendix~\ref{app:comp-extras}) and \eqref{eq:if-S}. The above display follows by plugging in the definitions of $\varphi_{q_a}$, rearranging, and because $\bbE \Big[ \one(A=b) - \widehat \pi_b(V) - \big\{ \pi_b(V) - \widehat \pi_b(V) \big\} \mid V \Big] = 0$. We consider each line of the above display in order.

    \bigskip

    First, by the proof of Proposition~\ref{prop:S-eif} and omitting arguments, 
    \begin{align*}
        &\bbE \Big\{  \widehat \varphi_S(Z) + \widehat S(V \in C_V) - S(V \in C_V)  \mid V \Big\}  \\
        &= \sum_{b=1}^{d}  \left( \bbE \Big[ \{ s^\prime(\widehat \pi_b) - s^\prime (\pi_b) \} (\pi_b - \widehat \pi_b) + \left\{ s^{\prime \prime}(\pi_b) + o(1) \right\} (\widehat \pi_b - \pi_b)^2 \mid V \Big] \right)\prod_{c \neq b}^{d} \widehat \bbE \{ s (\widehat \pi_c) \mid  V \} \\
        &+ \sum_{b=2}^{d} \left[ \bbE \{ s(\pi_b) \mid V \} - \widehat \bbE \big\{ s( \widehat \pi_b ) \mid V \big\} \right] \left[ \prod_{c < b} \widehat \bbE \{ s( \widehat \pi_c) \mid V \} - \prod_{c < b} \bbE \{ s(\pi_b) \mid V \} \right] \left[ \prod_{c > b} \widehat \bbE \{ s (\widehat \pi_c) \mid V \} \right].
    \end{align*}
    Therefore, by the boundedness assumption on $\widehat \bbE \{ \widehat \mu_b(X) \mid V \}$, H\"{o}lder's inequality, and Cauchy-Schwarz, \eqref{eq:decomp-1} satisfies
    \begin{align*}
        &\Bigg| \bbE \Bigg( \widehat \bbE \{ \widehat \mu_b(X) \mid V\} \big\{ \widehat \rho_a(A=b \mid V) - \widehat \pi_b(V) \big\} \bbE \Big\{  \widehat \varphi_S(Z) + \widehat S(V \in C_V) - S(V \in C_V)  \mid V \Big\} \Bigg) \Bigg| \\
        &\lesssim \sum_{b=1}^d \bbE \left[ \Big| \bbE \left\{ s^{\prime \prime} ( \pi_b ) (\widehat \pi_b - \pi_b )^2 \mid V \right\} \Big| \right] + \left\lVert \bbE( \widehat \pi_b - \pi_b \mid V ) \right\rVert \left\lVert \bbE \{ s^\prime(\widehat \pi_b) - s^\prime(\pi_b) \mid V \} \right\rVert \\
        &\hspace{0.5in}+ \left\lVert \widehat \bbE \{ s (\widehat \pi_b) \mid V \} - \bbE \{ s(\pi_b) \mid V \} \right\rVert \sum_{c < b} \left\lVert \widehat \bbE \{ s (\widehat \pi_c) \mid V \} - \bbE \{ s(\pi_c) \mid V \} \right\rVert  \\
        &=  \sum_{b=1}^d \left\lVert s^{\prime \prime}(\pi_b)^{1/2} (\widehat \pi_b - \pi_b) \right\rVert^2  + \left\lVert \bbE( \widehat \pi_b - \pi_b \mid V ) \right\rVert \left\lVert \bbE \{ s^\prime(\widehat \pi_b) - s^\prime(\pi_b) \mid V \} \right\rVert  \\
        &+  \sum_{b=1}^{d} \sum_{c < b} \left\lVert \widehat \bbE \{ s (\widehat \pi_b) \mid V \} - \bbE \{ s(\pi_b) \mid V \} \right\rVert \left\lVert \widehat \bbE \{ s (\widehat \pi_c) \mid V \} - \bbE \{ s(\pi_c) \mid V \} \right\rVert. 
    \end{align*}    
    Next, again by the boundedness assumption on $\widehat \bbE \{ \widehat \mu_b(X) \mid V \}$, H\"{o}lder's inequality, and Cauchy-Schwarz, \eqref{eq:decomp-2} satisfies
    \small
    \begin{align*}
        &\Bigg| \bbE \Bigg( \widehat \bbE \{ \widehat \mu_b(X) \mid V\}  \big\{ S(V \in C_V) - \widehat S(V \in C_V) \big\} \bbE \Big[ \big\{ \widehat \rho_a(A=b \mid V) - \widehat \pi_b(V) \big\} - \big\{ \rho_a(A=b \mid V) - \pi_b(V) \big\} \mid V \Big] \Bigg) \Bigg| \\
        &\lesssim \left[  \sum_{c=1}^{d}  \left\lVert \widehat \bbE \{ s(\widehat \pi_c) \mid V \} - \bbE \{ s(\pi_c) \mid V \} \right\rVert \right] \Bigg( \lVert \widehat \pi_b(V) - \pi_b(V) \rVert + \one(b \neq a) \left\lVert f\{ \widehat \pi_b(V) \} - f\{ \pi_b(V) \} \right\rVert \\
        &\hspace{3in}+ \one(b = a) \sum_{b \neq a}^d \left\lVert f \{ \widehat \pi_b(V) \} - f \{ \pi_b(V) \} \right\rVert \Bigg) 
    \end{align*}
    \normalsize
    \noindent  Finally, for \eqref{eq:decomp-3}, we have
    \begin{align*}
        \bbE \Big\{ &\widehat \varphi_{\rho_a}(Z; V, b) + \widehat \rho_a(A=b \mid V) - \rho_a(A=b \mid V) \mid V \Big\} = \\
        &\one(b \neq a) \Bigg( \Big[ f^\prime \{ \widehat \pi_b(V) \} - f^\prime \{ \pi_b(V) \} \Big] \{ \pi_b(V) - \widehat \pi_b(V) \} + \Big[ f^{\prime \prime} \{ \pi_b(V) \} + o(1) \Big] \left\{ \widehat \pi_b(V) - \pi_b(V) \right\}^2 \Bigg) \\
        &+ \one(b = a) \Bigg( \sum_{b \neq a} \Big[ f^\prime \{ \widehat \pi_b(V) \} - f^\prime \{ \pi_b(V) \} \Big] \{ \pi_b(V) - \widehat \pi_b(V) \} + \Big[ f^{\prime \prime} \{ \pi_b(V) \} + o(1) \Big] \left\{ \widehat \pi_b(V) - \pi_b(V) \right\}^2 \Bigg) 
    \end{align*}
    by iterated expectations on $V$ and a second order Taylor expansion of $f\{ \widehat \pi_b(V) \} - f\{ \pi_b(V) \}$.  Then, by the boundedness assumption on $\widehat \bbE \{ \widehat \mu_b(X) \mid V \}$, H\"{o}lder's inequality, and Cauchy-Schwarz, \eqref{eq:decomp-3} satisfies
    \begin{align*}
        \Bigg| \bbE \Bigg( \widehat \bbE \{ \widehat \mu_b(X) \mid V) \} \widehat S(V \in C_V) &\bbE \Big\{ \widehat \varphi_{\rho_a}(Z; V, b) + \widehat \rho_a(A=b \mid V) - \rho_a(A=b \mid V) \mid V \Big\} \Bigg) \Bigg| \\
        &\lesssim \one(b \neq a) \left \lVert f^\prime ( \widehat \pi_b ) - f^\prime ( \pi_b ) \right\rVert \left\lVert \widehat \pi_b - \pi_b \right\rVert + \bbE \big\{ \left| f^{\prime \prime} ( \pi_b ) \right| \left( \widehat \pi_b - \pi_b \right)^2 \big\} \\
        &+ \one(b = a) \sum_{b \neq a} \left \lVert f^\prime ( \widehat \pi_b ) - f^\prime ( \pi_b ) \right\rVert \left\lVert \widehat \pi_b - \pi_b \right\rVert + \bbE \big\{ \left| f^{\prime \prime} ( \pi_b ) \right| \left( \widehat \pi_b - \pi_b \right)^2 \big\}
    \end{align*}
    For the third term in~\eqref{eq:original-display}, we have, by Cauchy-Schwarz and the triangle inequality,
    \begin{align*}
        \Bigg| \bbE \bigg( &\Big\{ \widehat q_a(A=b \mid V) - q_a(A=b \mid V) \Big\} \Big[ \bbE \{ \mu_b(X) \mid V \} - \widehat \bbE \{ \widehat \mu_b(X) \mid V \} \Big] \bigg) \Bigg| \\
        &\lesssim \left\lVert \widehat \bbE \{ \widehat \mu_b(X) \mid V \} - \bbE \{ \mu_b(X) \mid V \} \right\rVert \Bigg( \left\lVert \widehat \pi_b(V) - \pi_b(V) \right\rVert + \sum_{c=1}^d \left\lVert \widehat \bbE \big[ s\{ \widehat \pi_c(X) \} \mid V \big] - \bbE \big[ s\{ \pi_c(X) \} \mid V \big] \right\rVert \\
        &\hspace{2.8in}+ \one(b = a) \lVert f(\widehat \pi_b) - f(\pi_b) \rVert + \one(b \neq a) \Big[ \sum_{b \neq a} \lVert f(\widehat \pi_b) - f(\pi_b) \rVert \Big] \Bigg).
    \end{align*}
    Bringing everything together, and re-introducing $V$ and $X$ arguments, we have
    \begin{align*}
        &\left| \bbE \left( \widehat \psi_a - \psi_a \right) \right| \lesssim \sum_{b=1}^{d} \left| \bbE \Bigg(  \bbE \left[ \left\{ \frac{\pi_b(X) - \widehat \pi_b(X)}{\widehat \pi_b(X)} \right\} \left\{ \mu_b(X) - \widehat \mu_b(X) \right\} \mid V \right] \widehat q_a(A=b \mid V) \Bigg) \right| \\
        &+ \sum_{l=1}^{d} \sum_{b=1}^d \left\lVert s^{\prime \prime}\{ \pi_b(X) \}^{1/2} \{ \widehat \pi_b(X) - \pi_b(X) \} \right\rVert^2  \\
        &\hspace{0.5in} + \left\lVert \bbE\{ \widehat \pi_b(X) - \pi_b(X) \mid V \} \right\rVert \left\lVert \bbE \{ s^\prime\{ \widehat \pi_b(X) \} - s^\prime \{ \pi_b(X) \} \mid V \} \right\rVert  \\
        &+ \sum_{l=1}^{d} \sum_{b=1}^{d} \sum_{c < b} \left\lVert \widehat \bbE \big[ s \{ \widehat \pi_b(X) \} \mid V \big] - \bbE \big[ s\{ \pi_b(X) \} \mid V \big] \right\rVert \left\lVert \widehat \bbE \big[ s \{ \widehat \pi_c(X) \} \mid V \big] - \bbE \big[ s\{ \pi_c(X) \} \mid V \big] \right\rVert \\
        &+ \sum_{b=1}^{d} \left[ \sum_{c=1}^{d}  \left\lVert \widehat \bbE \big[ s \{ \widehat \pi_c(X) \} \mid V \big] - \bbE \big[ s\{ \pi_c(X) \} \mid V \big] \right\rVert \right] \Bigg( \lVert \widehat \pi_b(V) - \pi_b(V) \rVert + \one(b \neq a) \left\lVert  f\{ \widehat \pi_b(V) \} - f \{ \pi_b(V) \} \right\rVert  \\
        &\hspace{3.5in}+ \one(b = a) \sum_{b \neq a}^d \left\lVert f\{ \widehat \pi_b(V) \} - f\{ \pi_b(V) \} \right\rVert \Bigg) \\
        &+ \sum_{b=1}^{d} \one(b \neq a) \left \lVert f^\prime \{ \widehat \pi_b(V) \} - f^\prime \{ \pi_b(V) \} \right\rVert \left\lVert \widehat \pi_b(V) - \pi_b(V) \right\rVert + \bbE \Big[ \left| f^{\prime \prime} \{ \pi_b(V) \} \right| \left\{ \widehat \pi_b(V) - \pi_b(V) \right\}^2 \Big] \\
        &+ \sum_{b=1}^{d} \one(b = a) \left \lVert f^\prime \{ \widehat \pi_b(V) \} - f^\prime \{ \pi_b(V) \} \right\rVert \left\lVert \widehat \pi_b(V) - \pi_b(V) \right\rVert + \bbE \Big[ \left| f^{\prime \prime} \{ \pi_b(V) \} \right| \left\{ \widehat \pi_b(V) - \pi_b(V) \right\}^2 \Big] \\
        &+ \sum_{b=1}^{d} \left\lVert \widehat \bbE \{ \widehat \mu_b(X) \mid V \} - \bbE \{ \mu_b(X) \mid V \} \right\rVert \Bigg( \left\lVert \widehat \pi_b(V) - \pi_b(V) \right\rVert + \sum_{c=1}^d \left\lVert \widehat \bbE \big[ s\{ \widehat \pi_c(X) \} \mid V \big] - \bbE \big[ s\{ \pi_c(X) \} \mid V \big] \right\rVert \\
        &\hspace{2in}+ \one(b \neq a) \lVert f\{ \widehat \pi_b(V) \} - f\{ \pi_b(V)\} \rVert + \one(b = a) \Big[ \sum_{b \neq a} \lVert f \{ \widehat \pi_b(V) \} - f\{ \pi_b(V) \} \rVert \Big] \Bigg).
    \end{align*}
    Next, note that, by Jensen's inequality,
    $$
    \lVert \bbE \{ \widehat f(Z) - f(Z) \mid V \} \rVert \leq \lVert \widehat f(Z) - f(Z) \rVert.
    $$
    Also notice that many of the outer sums over $l =1$ to $d$ can be simplified to just multiplying by a factor of $d$. We conclude
    \begin{align*}
        &\left| \bbE \left( \widehat \psi_a - \psi_a \right) \right| \lesssim \sum_{b=1}^{d} \left| \bbE \Bigg(  \bbE \left[ \left\{ \frac{\pi_b(X) - \widehat \pi_b(X)}{\widehat \pi_b(X)} \right\} \left\{ \mu_b(X) - \widehat \mu_b(X) \right\} \mid V \right] \widehat q_a(A=b \mid V) \Bigg) \right| \\
        &+ d \Bigg[ \sum_{b=1}^d \left\lVert s^{\prime \prime}\{ \pi_b(X) \}^{1/2} \{ \widehat \pi_b(X) - \pi_b(X) \} \right\rVert^2 + \left\lVert \bbE\{ \widehat \pi_b(X) - \pi_b(X) \mid V \} \right\rVert \left\lVert \bbE \{ s^\prime\{ \widehat \pi_b(X) \} - s^\prime \{ \pi_b(X) \} \mid V \} \right\rVert \Bigg]  \\
        &+ d \Bigg[ \sum_{b=1}^{d} \sum_{c < b} \left\lVert \widehat \bbE \big[ s\{ \widehat \pi_b(X) \} \mid V \big] - \bbE \big[ s\{ \pi_b(X) \} \mid V \big] \right\rVert \left\lVert \widehat \bbE \big[ s\{ \widehat \pi_c(X) \} \mid V \big] - \bbE \big[ s\{ \pi_c(X) \} \mid V \big] \right\rVert \Bigg] \\
        &+ \Bigg[ \sum_{b=1}^{d} \left\lVert \widehat \pi_b(X) - \pi_b(X) \right\rVert +  \sum_{b \neq a} \left\lVert f \{ \widehat \pi_b(X) \} - f \{ \pi_b(X) \} \right\rVert \Bigg] \left( \sum_{c=1}^{d}  \left\lVert \widehat \bbE \big[ s\{ \widehat \pi_c(X) \} \mid V \big] - \bbE \big[ s\{ \pi_c(X) \} \mid V \big] \right\rVert \right)  \\
        &+ \sum_{b \neq a}^{d} \left \lVert f^\prime \{ \widehat \pi_b(V) \} - f^\prime \{ \pi_b(V) \} \right\rVert \left\lVert \widehat \pi_b(V) - \pi_b(V) \right\rVert + \left\lVert f^{\prime \prime} \{ \pi_b(V) \}^{1/2} \{ \widehat \pi_b(V) - \pi_b(V) \} \right\rVert^2 \\
        &+ \sum_{b=1}^{d} \left\lVert \widehat \bbE \{ \widehat \mu_b(X) \mid V \} - \bbE \{ \mu_b(X) \mid V \} \right\rVert \left\lVert \widehat \pi_b(V) - \pi_b(V) \right\rVert \\
        &+ \sum_{b=1}^{d} \left\lVert \widehat \bbE \{ \widehat \mu_b(X) \mid V \} - \bbE \{ \mu_b(X) \mid V \} \right\rVert \Bigg( \sum_{c=1}^d \left\lVert \widehat \bbE \big[ s\{ \widehat \pi_c(X) \} \mid V \big] - \bbE \big[ s\{ \pi_c(X) \} \mid V \big] \right\rVert \Bigg) \\
        &+ \sum_{b=1}^{d} \left\lVert \widehat \bbE \{ \widehat \mu_b(X) \mid V \} - \bbE \{ \mu_b(X) \mid V \} \right\rVert \Bigg( \one(b \neq a) \lVert f\{ \widehat \pi_b(V) \} - f\{ \pi_b(V)\} \rVert \\
        &\hspace{3in} + \one(b = a) \Big[ \sum_{b \neq a} \lVert f \{ \widehat \pi_b(V) \} - f\{ \pi_b(V) \} \rVert \Big]  \Bigg).
    \end{align*}
\end{proof}

\bigskip
\noindent \textbf{Proof of Corollary~\ref{cor:normal}}

\normalsize
\begin{proof}
    First, note that Assumption~\ref{asmp:strong-positivity} allows both \eqref{eq:bias-convergence} and the second condition of Theorem~\ref{thm:dr-est} to hold simultaneously.  Then, 
    \begin{align*}
        \widehat \psi_a - \psi_a &= \bbP_n \left\{ \widehat \varphi_a(Z) \right\} - \bbE \left\{ \varphi_a(Z) \right\} \\
        &= (\bbP_n - \bbP) \left\{ \varphi_a(Z) \right\} + (\bbP_n - \bbP) \left\{ \widehat \varphi_a(Z) - \varphi_a(Z) \right\} + \bbP \left\{ \widehat \varphi_a(Z) - \varphi_a(Z) \right\}
    \end{align*} 
    where the second line follows by adding zero.  The first term satisfies the central limit theorem in the result.  Meanwhile, by Lemma 2 in \citet{kennedy2020sharp}, cross-fitting, and \eqref{eq:ifs-consistent}, $(\bbP_n - \bbP) \{ \widehat \varphi_a(Z) - \varphi_a(Z) \} = o_{\bbP}(n^{-1/2})$. Finally, the third term satisfies $\bbP \{ \widehat \varphi_a(Z) - \varphi_a(Z) \} = o_{\bbP}(n^{-1/2})$ by \eqref{eq:bias-convergence}.
\end{proof}

\subsection{Intermediate algebra} \label{app:intermediate-algebra}

In this section, we provide several steps of intermediate algebra that are necessary to establish the results in the previous section.

\begin{remark} \label{rem:pi-eif}
    Let 
    \begin{equation} \label{eq:if-pi}
        \varphi_{\pi}(Z;V,b) = \one(A=b) - \pi_b(V).
    \end{equation}
    Then, $\bbE \{ \varphi_{\pi}(Z; V, b) \mid V \} = 0$ and 
    $$
    \bbE \{ \overline \varphi_{\pi}(Z; V, b) + \overline \pi_b(V) - \pi_b(V) \mid V \} = 0.
    $$
\end{remark}

\begin{proposition} \label{prop:f-eif}
    Suppose $f(\cdot)$ is twice differentiable function with bounded second derivative $f^{\prime \prime}(\cdot)$ and non-zero first derivative $f^\prime(\cdot)$. Let
    \begin{equation} \label{eq:if-f}
        \varphi_{f}(Z; V, b) = f^\prime \{ \pi_b(V) \}  \left\{ \one (A=b) - \pi_b(V) \right\}
    \end{equation}
    Then, 
    $$
    \bbE \left[ \overline \varphi_{f}(Z; V, b)  + f\{ \overline \pi(V) \} - f\{ \pi(V)\} \mid V \right] =  \{ \pi(V) - \overline \pi(V) \}^2 \big[ f^{\prime \prime} \{ \overline \pi(V) \} + o(1) \big]
    $$
\end{proposition}

\begin{proof}
    Omitting $V$ arguments on the second line, 
    \begin{align*}
        \bbE \left[ \overline \varphi_{f}(Z; V, b)  + f\{ \overline \pi(V) \} - f\{ \pi(V)\} \mid V \right] &= f^\prime \{ \overline \pi(V) \} \left\{ \pi(V) - \overline \pi(V) \right\} + f\{ \overline \pi(V) \} - f\{ \pi(V)\} \\
        &= (\pi - \overline \pi)^2 \left\{ f^{\prime \prime}(\overline \pi) + o(1) \right\}
    \end{align*}
    where the first line follows by iterated expectations on $V$ and the second line by a second order Taylor expansion of $f(\overline \pi) - f(\pi)$.
\end{proof}

\begin{proposition}
    Suppose $V$ is discrete and $f(\cdot)$ is a twice differentiable function and let. Then,
    \begin{align*}
        \bbE \{ \overline \varphi_{\rho_a}(Z; V, b) &+ \overline \rho_a(A=b \mid V) - \rho_a(A=b \mid V) \mid V \} \\
        &=\one(b \neq a) \big\{ \pi_b(V) - \overline \pi_b(V) \big\}^2 \left[ f^{\prime \prime}\big\{ \overline \pi_b(V) \big\} + o(1) \right] \\
        &- \one(b=a) \sum_{b \neq a} \big\{ \pi_b(V) - \overline \pi_b(V) \big\}^2 \left[ f^{\prime \prime}\big\{ \overline \pi_b(V) \big\} + o(1) \right],
    \end{align*}
    where $\rho_a$ and $\varphi_{\rho_a}$ are defined in \eqref{eq:def-rho} (in Appendix~\ref{app:comp-extras}) and \eqref{eq:if-rho_a} (in the main paper).
\end{proposition}

\begin{proof}
    When $b\neq a$, this follows directly from the previous result.  Meanwhile, when $b \neq a$, notice that $\one(b=a)$ terms cancel, leaving $-\one(b=a) \left[ \sum_{b\neq a} \overline \varphi_f(Z; V, b) + f\{ \overline \pi_b(V) \} - f\{ \pi_b(V) \} \right]$.  The inner summands can be analyzed as in the previous result.
\end{proof}

\begin{proposition}
    Let
    \begin{equation} \label{eq:if-s}
        \varphi_{s}(Z; V, b) = s^\prime\{ \pi(X) \} \left\{ \one(A=b) - \pi_b(X) \right\} + s\{ \pi_b(X) \} - \bbE \big[ s\{ \pi_b(X) \} \mid V\big]. 
    \end{equation}
    Then,
    \begin{align*}
        &\bbE \left( \overline \varphi_s(Z; V, b) + \overline \bbE \big[ s\{ \overline \pi(X) \} \mid V \big] - \bbE \big[ s\{ \pi(X) \} \mid V \big] \mid V \right) \\
        &\hspace{1in} = \bbE \big( \big[ s^{\prime \prime} \{ \overline \pi(X) \} + o(1) \big] \{ \overline \pi(X) - \pi(X) \}^2 \mid V \big)
    \end{align*}

\end{proposition}

\begin{proof}
     Omitting $V$ arguments, we have
     \begin{align*}
        \bbE \bigg( &\overline \varphi_{s}(Z; V, b) + \overline \bbE \big[ s\{ \overline \pi(X) \} \mid V \big] - \bbE \big[ s\{ \pi(X) \} \mid V \big] \mid V \bigg)  \\
        &= \bbE \big[ s^\prime(\overline \pi) ( \pi - \overline \pi) \mid V  \big] + \bbE \{ s(\overline \pi) \mid V  \} - \overline \bbE \{ s(\overline \pi) \mid V \} + \overline \bbE \{ s(\overline \pi) \mid V \} - \bbE \{ s(\pi) \mid V  \}  \\
        &= \bbE \Big\{ s^\prime(\overline \pi)(\pi - \overline \pi) + s(\overline \pi) - s(\pi) \mid V \Big\}  \\
        &= \bbE \big[ \{ s^{\prime \prime}(\overline \pi) + o(1) \} (\overline \pi - \pi)^2 \mid V\big]
     \end{align*}
     where the first line follows by iterated expectations on $X$ and $V$, the second line by cancelling terms and iterated expectations on $V$, and the final line by a second order Taylor expansion of $s(\overline \pi) - s(\pi)$, where $s^{\prime \prime}$ is the second derivative of $s$.
\end{proof}

\begin{proposition} \label{prop:S-eif}
    Let $S(V \in C_V) = \prod_{a=1}^{d} \bbE \big[ s\{ \pi_a(X) \} \mid V \big]$ and  
    \begin{equation} \label{eq:if-S}
      \varphi_{S}(Z; V) = \sum_{a=1}^{d}  \varphi_{s}(Z; V, a) \prod_{b \neq a}^{d} \bbE \big[ s\{ \pi_b(X) \} \mid V  \big],
    \end{equation}
    and $\varphi_{s}(Z; V, b)$ is defined in \eqref{eq:if-s}.  Then, omitting $X$ arguments on the right-hand side, 
    \begin{align*}
        &\bbE \left( \varphi_{S}(Z; V) + \prod_{a=1}^{d} \overline \bbE \big[ s \{ \overline \pi_a(X) \} \mid V \big] - \prod_{a=1}^{d} \bbE \big[ s \{ \pi_a(X) \} \mid V \big] \mid V \right) \\
        &= \sum_{a=1}^{d} \bbE \Big[ \{ s^\prime(\overline \pi_a) - s^\prime (\pi_a) \} (\pi_a - \overline \pi_a) + \left\{ s^{\prime \prime}(\pi_a) +o(1) \right\} (\overline \pi_a - \pi_a)^2  \mid V \Big] \prod_{b \neq a}^{d} \overline \bbE \{ s (\overline \pi_b) \mid  V \}  \\
        &\hspace{0.2in}+ \sum_{a=2}^{d} \left[ \bbE \{ s(\pi_a) \mid V \} - \overline \bbE \big\{ s( \overline \pi_a ) \mid V \big\} \right] \left[ \prod_{b < a} \overline \bbE \{ s( \overline \pi_b) \mid V \} - \prod_{b < a} \bbE \{ s(\pi_b) \mid V \} \right] \left[ \prod_{b > a} \overline \bbE \{ s (\overline \pi_b) \mid V \} \right].
    \end{align*}
\end{proposition}

\begin{proof}
    Omitting arguments, we have 
    \begin{align*}
        &\bbE \left( \varphi_{S}(Z; V) + \prod_{a=1}^{d} \overline \bbE \big[ s \{ \overline \pi_a(X) \} \mid V \big] - \prod_{a=1}^{d} \bbE \big[ s \{ \pi_a(X) \} \mid V \big] \mid V \right) \\
        &= \bbE \left( \sum_{a=1}^{d}  \Big[ s^\prime( \overline \pi_a ) \left\{ \one(A=a) - \overline \pi_a \right\} + s( \overline \pi_a ) - \overline \bbE \big\{ s( \overline \pi_a ) \mid V \big\} \Big] \prod_{b \neq a}^{d} \overline \bbE \{ s (\overline \pi_b) \mid V \} \mid V \right) \\
        &\hspace{0.2in}+  \prod_{a=1}^{d} \overline \bbE \big\{ s ( \overline \pi_a ) \mid V \big\} - \prod_{a=1}^{d} \bbE \big\{ s ( \pi_a ) \mid V \big\} \\
        &= \bbE \left( \sum_{a=1}^{d}  \Big[ s^\prime( \overline \pi_a ) \left( \pi_a - \overline \pi_a \right) + s( \overline \pi_a ) - \overline \bbE \big\{ s( \overline \pi_a ) \mid V \big\} \Big] \prod_{b \neq a}^{d} \overline \bbE \{ s (\overline \pi_b) \mid V \} \mid V \right) \\ 
        &\hspace{0.2in}+  \prod_{a=1}^{d} \overline \bbE \big\{ s ( \overline \pi_a ) \mid V \big\} - \prod_{a=1}^{d} \bbE \big\{ s ( \pi_a ) \mid V \big\} \\ 
        &=  \sum_{a=1}^{d} \bbE \Big[ s^\prime( \overline \pi_a ) \left( \pi_a - \overline \pi_a \right) + s( \overline \pi_a ) - s(\pi_a) \mid V \Big] \prod_{b \neq a}^{d} \overline \bbE \{ s (\overline \pi_b) \mid V \} \\
        &+ \sum_{a=1}^{d} \Big[ \bbE \big\{ s(\overline \pi_a) \mid V) \big\} - \overline \bbE \big\{ s( \overline \pi_a ) \mid V \big\} \Big] \prod_{b \neq a}^{d} \overline \bbE \{ s (\overline \pi_b) \mid V \}  + \prod_{a=1}^{d} \overline \bbE \big\{ s ( \overline \pi_a ) \mid V \big\} - \prod_{a=1}^{d} \bbE \big\{ s ( \pi_a ) \mid V \big\} \\
        &= \sum_{a=1}^{d}  \bbE \Big[ \{ s^\prime(\overline \pi_a) - s^\prime (\pi_a) \} (\pi_a - \overline \pi_a) + \left\{ s^{\prime \prime}(\pi_a) +o(1) \right\} (\overline \pi_a - \pi_a)^2  \mid V \Big] \prod_{b \neq a}^{d} \overline \bbE \{ s (\overline \pi_b) \mid  V \}  \\
        &\hspace{0.2in}+ \sum_{a=1}^{d} \left[ \bbE \{ s(\pi_a) \mid  V \} - \overline \bbE \big\{ s( \overline \pi_a ) \mid  V \big\} \right]  \prod_{b \neq a}^{d} \overline \bbE \{ s (\overline \pi_b) \mid V \} + \prod_{a=1}^{d} \overline \bbE \big\{ s ( \overline \pi_a ) \mid  V \big\} - \prod_{a=1}^{d} \bbE \big\{ s ( \pi_a ) \mid  V \big\}
    \end{align*}
    where the first line follows by definition, the second by iterated expectations on $X$ and iterated expectations on $V$, the third by adding and subtracting $\bbE \{ s(\pi) \mid  V \}$, and the fourth by a second order Taylor expansion of $s(\overline \pi) - s(\pi)$.  Finally, the results follows because Proposition~\ref{prop:algebra} implies that the final line of the above display satisfies
    \begin{align*}
        &\sum_{a=1}^{d} \left[ \bbE \{ s(\pi_a) \mid  v\} - \overline \bbE \big\{ s( \overline \pi_a ) \mid V \big\} \right]  \prod_{b \neq a}^{d} \overline \bbE \{ s (\overline \pi_b) \mid V \} + \prod_{a=1}^{d} \overline \bbE \big\{ s ( \overline \pi_a ) \mid V \big\} - \prod_{a=1}^{d} \bbE \big\{ s ( \pi_a ) \mid V \big\} \\
        &= \sum_{a=2}^{d} \left[ \bbE \{ s(\pi_a) \mid V \} - \overline \bbE \big\{ s( \overline \pi_a ) \mid V \big\} \right] \left[ \prod_{b < a} \overline \bbE \{ s( \overline \pi_b) \mid V \} - \prod_{b < a} \bbE \{ s(\pi_b) \mid V \} \right] \left[ \prod_{b > a} \overline \bbE \{ s (\overline \pi_b) \mid V \} \right].
    \end{align*}
\end{proof}

\begin{proposition} \label{prop:algebra}
    Suppose $\{a_j\}_{j=1}^{d}$ and $\{b_j \}_{j=1}^d$ are two sequences.  Then 
    $$
    \sum_{j=1}^{d} (b_j - a_j) \left( \prod_{l \neq j} a_l \right) + \prod_{l=1}^{d} a_l - \prod_{l=1}^{d} b_l = \sum_{j=2}^{d} (b_j - a_j) \left( \prod_{l < j} a_l - \prod_{l < j} b_l \right) \left( \prod_{l > j} a_l \right).
    $$
\end{proposition}

\begin{proof}
    Let $P(d) = \sum_{j=1}^{d} (b_j - a_j) \left( \prod_{l \neq j} a_l \right) + \prod_{l=1}^{d} a_l - \prod_{l=1}^{d} b_l$.  Notice that
    $$
    P(d) - a_d P(d-1) = (b_d - a_d) \left( \prod_{l=1}^{d-1} a_l - \prod_{l=1}^{d-1} b_l \right).
    $$
    Notice also that $P(2) = (b_2 - a_2)(b_1 - a_1)$. Then, rearranging this recursive definition and plugging in $P(2)$ in the final summand yields the result.
\end{proof}

\end{document}